\documentclass[twoside]{article}

\usepackage[accepted]{aistats2021}
\usepackage{lmodern}
\usepackage[T1]{fontenc}
\usepackage{relsize}

\usepackage[square,authoryear]{natbib}
\usepackage[colorlinks=true,linkcolor=black,citecolor=blue]{hyperref}
\usepackage{tikz}
\usepackage{amsfonts}
\usepackage[utf8]{inputenc}
\usepackage[english]{babel}
\usepackage{amsthm}
\usepackage{algorithm}
\usepackage{algorithmic}
\usepackage{bbold}
\usepackage{dsfont}

\usetikzlibrary{arrows}
\usepackage[nolabel, final]{showlabels}
\usepackage{graphicx}
\usepackage{subfigure}
\usepackage{centernot}
\usepackage{mathtools}
\usepackage{bbm}

\usepackage{wrapfig}
\usepackage{enumerate}

\newtheorem{Theorem}{Theorem}

\newtheorem{Corollary}{Corollary}
\newtheorem{Definition}{Definition}
\newtheorem{Proposition}{Proposition}
\newtheorem{Assumption}{Assumption}

\newcommand{\R}{\mathbb{R}} 
\renewcommand{\H}{\mathcal{H}} 
\newcommand{\G}{\mathcal{G}} 
\newcommand{\X}{\mathcal{X}} 
\newcommand{\A}{\mathcal{T}} 
\renewcommand{\P}{\mathbb{P}} 
\newcommand{\E}{\mathbb{E}} 
\newcommand{\N}{\mathbb{N}} 
\newcommand{\T}{\mathcal{T}}
\newcommand{\Trans}{\hspace{-0.25ex}\top\hspace{-0.25ex}}
\newcommand{\Var}{\mathbb{V}ar}
\newcommand{\gKSS}{\operatorname{gKSS}}

\newcommand{\wk}[1]{\textcolor{orange}{#1}}

\newcommand{\gr}[1]{\textcolor{magenta}{#1}}

\newcommand{\IR}{\mathbbm{R}}

\def\%#1{\mathcal{#1}}

\newcommand{\ahalf}{{\textstyle\frac{1}{2}}}

\newcommand{\beas}{\begin{eqnarray*}}
\newcommand{\enas}{\end{eqnarray*}}

\newcommand{\bea}{\begin{eqnarray}}
\newcommand{\ena}{\end{eqnarray}}

\newcommand{\IE}{\mathbbm{E}}
\newcommand{\IP}{\mathbbm{P}}

\newcommand{\Id}{\mathop{\mathrm{Id}}}

\newtheorem{theorem}{Theorem}[section]

\def\be#1\ee{\begin{equation*}#1\end{equation*}}
\def\ben#1\ee{\begin{equation}#1\end{equation}}

\def\bes#1\ee{\begin{equation*}\begin{split}#1\end{split}\end{equation*}}
\def\besn#1\ee{\begin{equation}\begin{split}#1\end{split}\end{equation}}

\def\bg#1\ee{\begin{gather*}#1\end{gather*}}
\def\bgn#1\ee{\begin{gather}#1\end{gather}}

\def\bm#1\ee{\begin{multline*}#1\end{multline*}}
\def\bmn#1\ee{\begin{multline}#1\end{multline}}

\def\ba#1\ee{\begin{align*}#1\end{align*}}
\def\ban#1\ee{\begin{align}#1\end{align}}

\def\bbklr#1{\Bigl(#1\Bigr)}

\def\norm#1{\Vert#1\Vert}
\def\bnorm#1{\bigl\Vert#1\bigr\Vert}

\def\abs#1{\vert#1\vert}
\def\babs#1{\bigl\vert#1\bigr\vert}

\begin{document}

%

%

\twocolumn[

\aistatstitle{A Stein Goodness-of-fit Test for Exponential Random Graph Models}

\aistatsauthor{ Wenkai Xu \And Gesine Reinert }

\aistatsaddress{ Gatsby Computational Neuroscience Unit\\ University College London \And    Department of Statistics\\
  Oxford University } ]

\begin{abstract}
We propose and analyse a novel nonparametric goodness-of-fit testing procedure for exchangeable exponential random graph models (ERGMs) when a single network realisation is observed. The test determines how likely it is that the observation is generated from a target unnormalised ERGM density.  Our test statistics are derived from a kernel Stein discrepancy, a divergence constructed  via Stein’s method using functions {in}
a reproducing kernel Hilbert space, 
combined with a discrete Stein operator for ERGMs.  The test is a Monte Carlo test {based on}
simulated networks from the target ERGM. We show theoretical properties for the testing procedure for a class of ERGMs. Simulation studies and real network applications are presented.
\end{abstract}

\section{INTRODUCTION}
{Complex data from many application areas are often represented as networks, and probabilistic network models {help}
to understand the expected behaviour of such networks.} 
{In social science, exponential random graph models (ERGMs) have been successfully employed for this task, see for example \cite{wasserman1994social}.} 
{ERGMs can be viewed as}
exponential family models or energy-based models, {and as typical for such models,} statistical inference 
{for} ERGMs
suffers from intractable normalisation constants. 
Monte Carlo methods for parameter estimations in ERGMs {alleviate this issue \citep{snijders2002markov}}, 
and model diagnoses via maximum likelihood (MLE) and maximum pseudo-likelihood
are developed \citep{morris2008ergm}. Statistical properties with particular attention to the normalisation constant {are} 
studied {in}  \cite{chatterjee2013estimating}. 
{For analysing distributions with intractable normalisation constants, Stein's method \citep{barbour2005introduction, chen2010} provides a promising approach \citep{chwialkowski2016kernel,liu2016stein, bresler2019stein}. In \cite{reinert2019approximating}, Stein's method is developed for ERGMs {but not yet applied to goodness-of-fit tests}.  
} 

{Goodness}-of-fit tests for random graph models address the problem of whether the proposed model generates the observed network(s), and 
play {a key role} in understanding and interpreting network structures in real-world applications. 
{A main issue is that replicates are usually not available; the data are represented as only one network.} 
Standard goodness-of-fit tests for ERGMs to date rely on Monte Carlo tests for particular summary statistics
such as vertices degrees \citep{ouadah2020degree}, motifs or subgraph counts \citep{bhattacharyya2015subsampling, ospina2019assessment, chen2019bootstrap}, 
{or spectral properties \citep{shore2015spectral}}.   
The goodness-of-fit tests for ERGM  in \cite{hunter2008goodness} or \cite{ schweinberger2012statistical}  {also assess} the model assumptions via  graphical assessments. 
{\cite{lospinoso2019goodness} combines such statistics into a Mahalanobis-type distance which is assessed via Monte Carlo tests. 
The consistency of type 1 error and the power of the test have not yet been {systematically} investigated.

{N}onparametric goodness-of-fit tests based on 
Stein operators \citep{gorham2015measuring, ley2017stein} and functions in {a}
reproducing kernel Hilbert space (RKHS) \citep{RKHSbook} {for data with replicates build} 
on a kernel Stein discrepancy (KSD) that utilises the strength of a Stein operator to {treat} unnormalised models and optimises {over} test functions in a rich enough RKHS to best distinguish the data from the model distributions.
{S}uch schemes are consistent and have high test power in various scenarios, including multivariate distributions \citep{chwialkowski2016kernel, liu2016kernelized}, discrete distributions \citep{yang2018goodness}, point processes \citep{yang2019stein}, directional distributions \citep{xu2020stein}, and censored data \citep{tamara2020kernelized}. {These scenarios are typically based on i.i.d. samples from the distributions.} 
In addition, the properties of kernel mean embeddings \citep{RKHSbook, muandet2017kernel} enable the extraction of distributional features to perform model comparison and model criticism \citep{jitkrittum2017linear,jitkrittum2018informative, kanagawa2019kernel, jitkrittum2020testing}. 

{Here} we propose a novel goodness-of-fit testing procedure for ERGMs combining a Stein {operator}  {for ERGMs} and functions in 
an 
RKHS. {The class of ERGMs treated here are undirected networks which, when the number of vertices tends to infinity, can be approximated by a suitably chosen Bernoulli random graph, with edge {probability} parameter that 
generally {does not equal} 
the MLE.} {The test is based on only \emph{one} observed network and 
{estimates} the Stein operator 
through re-sampling edge indicators.
{This test} compares the test statistics from {one} observed network to the simulated distribution of the statistics under the null model. As the Stein operator
characterises the target distribution {for this class of ERGMs, under a  member of this class serving as null hypothesis} {we derive 
theoretical results for the test statistic.}
{We also provide a theoretical justification of the proposed re-sampling procedure.} 

{To assess the performance of the test, we use simulated data as well as three real-world applications: Lazega's lawyer network \citep{lazega2001collegial}, a teenager friendship network \citep{steglich2006applying}, and a larger network of legislation co-sponsorship \citep{fowler2006connecting, fowler2006legislative}.}
{We find that on synthetic data, our test is more reliable and has higher power than the standard tests even when only a small number of edge indicators is re-sampled. Moreover, the test can be applied to networks on more vertices than its competitor tests. For the lawyer network, we confirm the suggestion by \cite{lazega2001collegial} of a Bernoulli random graph; for the friendship network we do not reject an ERGM with edges, two-stars and triangles as statistics. For the co-sponsorship network we do not reject a Bernoulli random graph fit whereas the model proposed in \cite{schmid2017exponential} is rejected at level $\alpha = 0.05$.}

{The paper is structured as follows.} We begin our presentation with a short review on ERGM, KSD and the ERGM Stein operator in Section \ref{sec:background}. Section \ref{sec:stein_operator} introduces {our} re-sampling Stein operator for ERGM,  Our goodness-of-fit testing procedure {which is based on what we call the {\it graph kernel Stein statistic} (gKSS), and  the relevant theoretical results}, are {given} in Section \ref{sec:gof}. In Section \ref{sec:exp}, we illustrate the test performances on synthetic data as well as real network applications. 

All the {proofs are deferred to the Supplementary Material. The Supplementary Material also contains more discussions, details for experiment settings and additional experimental results, {as well as a detailed comparison to the correesponding test in \cite{yang2018goodness}.}
{The code and data sets for the experiments are available at \url{https://github.com/clemon13/gkss.git}.} 

\section{BACKGROUND}\label{sec:background}

\subsection{Exponential Random Graph Model{s}} 
Exponential random graph models (ERGM)  are frequently used as parametric models for {social} network analysis \citep{wasserman1994social, holland1981exponential, frank1986markov}; 
{they include Bernoulli random graphs as well as stochastic blockmodels as special cases}. 
Here we restrict attention to undirected, unweighted simple graphs {on $n$ vertices}, 
{without self-loops or multiple edges.}
To define {such} an ERGM, we introduce the following notations.

Let $\G^{lab}_n$ be a set of vertex-labeled  graphs on $n$ vertices and,
for {$N =n(n-1)/2 $,}
{encode} $x \in \G^{lab}_n$ by an ordered collection of $\{0,1\}$ valued variables $x = (x_{ij})_{1 \le  i < j \le n} \in \{0,1\}^N$ {where} $x_{ij}=1$ {if and only if}
there is an edge between $i$ and $j$. 
{For a graph $H$ on at most $n$ vertices, let $V(H)$ denote the vertex set, 
and for $x\in\{0,1\}^N$, denote by $t(H,x)$  the number of
{\it edge-preserving} injections from $V(H)$ to $V(x)$; an injection $\sigma$ preserves edges if for all edges $vw$ of $H$ {with  $\sigma(v)<\sigma(w)$}, $x_{\sigma(v)\sigma(w)}=1$.
For  $v_H =| V(H)| \ge 3$  set 
$$
t_H(x)=\frac{t(H,x)}{n(n-1)\cdots(n-v_H+3)}.
$$
If $H{=H_1}$ is a single edge, then $t_H(x)$ is twice the number of edges of~$x$. In the exponent this scaling of counts matches \cite[Definition~1]{bhamidi2011mixing} and \cite[Sections~3 and~4]{chatterjee2013estimating}.
} 
An ERGM {
for the collection  $x\in \{0,1\}^{N}$
can be defined
} as follows, see  \cite{reinert2019approximating}.

\begin{Definition}
\label{def:ergm} 
Fix $n\in \N$ and $k \in \N$. {{L}et $H_1$ be a single edge {and f}or $l={2}, \ldots, k$ let}  $H_l$ be a connected graph on at most $n$ vertices; 
set $t_l(x) = t_{H_l}(x)$. For  $\beta = (\beta_1, \dots, \beta_k)^{\Trans} {\in \R^k}$ 
and
$t(x) =(t_1(x),\dots,t_k(x))^{\Trans} \in \R^k$ 
$X\in \G^{lab}_n$ follows  the exponential random graph model  $X\sim \operatorname{ERGM}(\beta, t)$ if for  $\forall x\in \G^{lab}_n$,
$$
    \P(X = x) = \frac{1}{\kappa_n(\beta)}\exp{\left(\sum_{l=1}^{k} \beta_l t_l(x) \right)}.
$$
Here $\kappa_n(\beta)$ is the normalisation constant.
\end{Definition}
{
The vector $\beta \in \R^k$ is the parameter vector and the statistics $t(x) =(t_1(x),\dots,t_k(x))^{\Trans} \in \R^k$  are sufficient statistics. Moreover, exchangeability holds; letting $[n] := \{1,\dots, n\}$, for any permutation  $\sigma:{[N] \to [N]}$,
}
$\P(x_1,\dots,x_N) = \P(x_{\sigma(1)},\dots,x_{\sigma(N)}).$    

Many random graph models can be set in this framework. The simplest example is the Bernoulli random graph (ER graph) {with edge probability $0 < p < 1$; in this case},  {$l=1$ and $H_1$ is a single edge}.
{ERGMs can use other statistic in addition to subgraph counts, and many ERGMs model directed networks.} 
Moreover, ERGM{s} can  model network with covariates such as using dyadic statistics to model group interactions between vertices \citep{hunter2008goodness}. 
{Here we restrict attention to the case which is treated 
in \cite{reinert2019approximating} {because}  it is for this case that a Stein characterization {is available}.
}

As the network size increases, the number of possible network configurations increases exponentially {in the number of possible edges}, making the normalisation constant $\kappa_n(\beta)$ {usually} prohibitive to compute in closed form.
Statistical inference on ERGM mainly relies on  MCMC type  methods that utilise the density ratio between proposed state and current state, where the normalisation constant cancels.

\subsection{Kernel Stein Discrepancies}\label{sec:ksd}

We briefly review the {notion of} kernel Stein discrepancy (KSD) {for continuous distributions}  \citep{gorham2015measuring,ley2017stein} and its associated statistical test \citep{chwialkowski2016kernel,liu2016kernelized}. 

Let $q$ be a smooth probability density on 
$\mathbb{R}^d$ that vanishes at the boundary.  The operator $\T_q: (\mathbb{R}^d \to \mathbb{R}^d) \to (\mathbb{R}^d \to \mathbb{R})$ is called a {\it Stein operator} if the following {\it Stein's identity} holds: $\E_q[{\T}_q f] = 0$, where $f=(f_1,\dots,f_d):\mathbb{R}^d \to \mathbb{R}^d$ is any bounded smooth function. The Stein operator $\mathcal{T}_q$ 
for continuous density \citep{chwialkowski2016kernel,liu2016kernelized} is defined as
\begin{align}
\mathcal{T}_q f(x)&=\sum_{i=1}^d \left( f_i(x) \frac{\partial}{\partial x^i} \log q(x) + \frac{\partial}{\partial x^i} f_i(x) \right)
.
\label{eq:steinRd}
\end{align}
{This Stein operator is also called Lagenvin-diffusion Stein operator \citep{barp2019minimum}.}
Since $q$ is assumed to vanish at the boundary and $f$ is bounded, the Stein identity holds due to integration by parts.
As the Stein operator $\mathcal{T}_q$ only requires the derivatives of $\log q$ {and thus}  does not involve computing the normalisation constant of $q$, {it is} useful for dealing with unnormalised models \citep{hyvarinen2005estimation}. 

A  {suitable class of functions ${\mathcal F}$ is such that if} 
 $\E_p[{\T}_q f] = 0$ for all functions $f { \in {\mathcal F}}$, then $p=q$ follows. 
{It is convenient to take 
${\mathcal F} = B_1( {\mathcal H })$,} the 
unit ball of {a} {large enough} RKHS { $ {\mathcal H }$}.
{In particular, t}he kernel Stein discrepancies (KSD) between two densities $p$ and $q$ {based on $\T_q$} is defined as
\begin{equation}
\operatorname{KSD}(p\|q, {{\mathcal H }}) =\sup_{f \in B_1(\mathcal H)} \mathbb{E}_{p}[{\T}_q f]. 
\label{eq:ksd}
\end{equation}
{Under mild regularity conditions,}  $\mathrm{KSD}(p\|q, {{\mathcal H }}) \geq 0$ and $\mathrm{KSD}(p\|q, {{\mathcal H }}) = 0$ if and only if $p=q$
\citep{chwialkowski2016kernel}, {making} 
KSD
a 
proper
discrepancy measure between probability densities. 

{The KSD in  Eq.\eqref{eq:ksd} can be used for testing the model goodness-of-fit as follows.} 
{One can show that}
$\operatorname{KSD}^2(p\|q, {{\mathcal H }}) = {\E}_{x,\tilde{x} \sim p} [h_q(x,\tilde{x})]$,
where {$x$ and $\tilde{x}$ are independent random variables with density $p$} and  $h_q(x,\tilde{x})$ {is given in explicit form {which does not involve $p$};
{\begin{eqnarray}   
h_q(x,\tilde{x})&=& \sum_{i=1}^d \left\langle  \frac{\partial \log q(x)}{\partial x^i}k(x,\cdot) + \frac{\partial k(x,\cdot)}{\partial x^i} , \right. \nonumber \\
&&\left. \frac{\partial \log q(\tilde{x})}{\partial \tilde{x}^i}k(\tilde{x},\cdot) +  \frac{\partial k(\tilde{x},\cdot)}{\partial \tilde{x}^i}\right\rangle_{\H}. \label{hform} 
\end{eqnarray}}  
Suppose we have a set of samples  $\{ x_1,\dots,x_n \}$ from an unknown density $p$ on $\mathbb{R}^d$ and the goodness-of-fit test aims to check whether $p=q$. 
Then $\mathrm{KSD}^2(p\|q, {{\mathcal H }})$ can be empirically estimated by {independent} samples from $p$ using a U-statistics or V-statistics.
The critical value is determined by  bootstrap based on
{weighted chisquare approximations for} 
U-statistics or V-statistics. 

{For goodness-of-fit test of discrete distributions}, 
\cite{yang2018goodness}  proposed a kernel discrete Stein discrepancy (KDSD). 
Essentially, the differential operator in Eq.\eqref{eq:steinRd} is replaced by an appropriately defined difference operator. KDSD is a useful method for assessing  goodness-of-fit of  ERGMs (as discrete random objects) when a large set of networks are observed \cite[Figure~1(d)]{yang2018goodness}, {but} {is not applicable} 
when only one single network 
is observed. More details
can be found in the Supplementary Material \ref{supp:KDSD_supp}.

\subsection{The {ERGM Stein Operator}}
 Instead of using the Stein operator from \cite{yang2018goodness} we employ the {Stein operator} from \cite{reinert2019approximating}. With 
 {$N= n(n-1)/2$}
 let $e_s \in  \{0,1\}^N$ be a vector with $1$ in 
 coordinate {$s$} and 0 in all others; 
$x^{(s,1)} = x + (1-x_s) e_s$ {has the $s$-entry  replaced of $x$ by the value 1, and} $x^{(s,0)} = x - x_s e_s$ {has  the $s$-entry  of $x$ replaced by the value 0; moreover,}
${x}_{- s}$ is the set of edge indicators with  entry $s$ removed.
Then a (Glauber dynamics) Markov process on $ \{0,1\}^N$ is introduced with transition probabilities 
$$\P (x \rightarrow x^{(s,1)} ) = \frac1N -
 \P( x \rightarrow  x^{(s,0)}) = \frac1N q_X(x^{(s,1)} | x)$$
where
$ q_X(x^{(s,1)} | {x_{-s}} ) = \P ( X_s=1| {X_{-s} = x_{-s}}).$
For the  ERGM$(\beta, t)$ from Definition \ref{def:ergm}, 
\begin{eqnarray*}
\lefteqn{q(x^{(s,1)}| {x_{-s}} ):= \exp\left\{\sum_{\ell=1}^k \beta_\ell t_\ell(x^{(s,1)})\right\} \times } \\ && \hspace{-4mm} 
\left( \exp\left\{ \sum_{\ell=1}^k \beta_\ell t_\ell(x^{(s,1)})\right\} +\exp\left\{\sum_{\ell=1}^k \beta_\ell t_\ell(x^{(s,0)})\right\} \right)^{-1} 
\end{eqnarray*}
and similarly the probability of the new edge being absent exchanges $1$ and $0$ in this formula
to give $q(x^{(s,0)}|{x_{-s}} )$. 
For
 $h: \{0,1\}^N \rightarrow \R$ let
$$\Delta_s h(x) = h( x^{(s,1)}) - h(x^{(s,0)}).$$
The generator $\T_{\beta, {t}} $ of this Markov process is the desired Stein operator and its expression simplifies to 
\begin{eqnarray} 
\T_{\beta, {t}} f(x)
	&=&  \frac{1}{N} \sum_{s\in[N]} \A^{(s)}_q f(x) 
	\label{eq:ergm_stein} 
	\end{eqnarray}
	with
	\begin{eqnarray}\label{eq:stein_component}
    \A^{(s)}_q f(x) &=& q(x^{(s,1)}|{x_{-s}} ) \Delta_s f(x) \nonumber  \\
    && + \left( f(x^{(s,0)}) - f(x)\right).
\end{eqnarray}
	When the ERGM is such that the  Markov process is  irreducible, then its stationary distribution is unique, and  if 
$\E_p[{\T}_{\beta, t}  f] = 0$ for all smooth test functions $f$, then $p$ is the distribution of ERGM$(\beta, t)$. Thus, the Stein operator characterises  ERGM$(\beta, t)$.
{{Moreover, for each $s \in [N],$
\begin{eqnarray}\label{meanzeroproperty} 
\mathbb{E}_{q} \T_{q}^{(s)}f = 0.
\end{eqnarray}
To see this,}  write 
\begin{eqnarray*}
\mathbb{E}_{q} \T_{q}^{(s)}f 
&= \sum_x {q(x_{-s})}  &\left(q(x^{(s,1)}| {x_{-s}} ) \T_{q}^{(s)}f(x^{(x,1)}) \right.\\ &&\left.  +q(x^{(s,0)}| {x_{-s}} )\T_{q}^{(s)}f(x^{(x,0)}) \right) .
\end{eqnarray*}
Substituting $x^{(s,1)}$ and $x^{(s,0)}$ in  \eqref{eq:stein_component} gives 
\begin{eqnarray*}
\lefteqn{q(x^{(s,1)}| {x_{-s}} )\T_{q}^{(s)}f(x^{(s,1)}) +q(x^{(s,0)}| {x_{-s}} )\T_{q}^{(s)}f(x^{(s,0)}) }\\
&=&q(x^{(s,1)}| {x_{-s}} )q(x^{(s,0)}| {x_{-s}} )\Big( f(x^{(s,0)}) - f(x^{(s,1)}) \\
&& + f(x^{(s,1)}) - f(x^{(s,0)})\Big)=0
\end{eqnarray*}
and Eq.\eqref{meanzeroproperty} follows.
} 


{Next}, we 
{introduce  {our}}  kernel Stein statistic for testing {the goodness-of-fit of an ERGM  based on a} single observed
network 
{as well as an estimator for it which is based on  re-sampling of edge indicators}.

\section{KERNEL STEIN STATISTICS from RE-SAMPLING}\label{sec:stein_operator}

\paragraph{Kernel Stein Statistics}
Based on the Stein operator representation Eq.\eqref{eq:ergm_stein}, we develop the kernel Stein statistics (KSS)\footnote{We 
avoid calling it a \emph{discrepancy}
since  our expectation is not taken over all ERGM samples as described in 
\cite{yang2018goodness}, {but instead based on a single network.}} for ERGMs. Similar to KSD in Eq.\eqref{eq:ksd}, we use the functions in a unit ball of an RKHS  {${\mathcal H}$}
as test functions.

The Stein operator in Eq.\eqref{eq:ergm_stein} can be written as expectation over edge variables $S$ with uniform probability $\P(S=s)=\frac{1}{N}, \forall s {\in [N]:=\{1, \ldots, N\}}$, independently of $x$,  namely 
\begin{align}
\A_q f(x) &= \sum_{s\in [N]} \P(S=s)\A^{(s)}_q f(x) =: {\E_S [ \A^{(S)}_q f(x)]} . \nonumber
\end{align}
Note that the expectation is taken over
$S$,  with the network  $x$ fixed {except for the coordinate $S$}.


After algebraic manipulations, Eq~\eqref{eq:stein_component} has the form
\begin{eqnarray}
 \lefteqn{\A_q^{(s)} f(x)}  \nonumber \\
 & \hspace{-0.2cm} = & \hspace{-0.13cm} q(x^{(s,1)}|{x_{-s}} )f (x^{(s,1)}) + q(x^{(s,0)}|{x_{-s}} )  f (x^{(s,0)})  - f(x) \nonumber\\
 & \hspace{-0.2cm} = & \hspace{-0.13cm} \E_{\{0,1\}}[f(X_{s},x_{-s}) | {x_{-s}} )] - f(x).\label{eq:stein_conditional_difference}
\end{eqnarray}
{Here $ \E_{\{0,1\}}$ refers to the expectation taken only over the value which ${X}_s$ takes on.}
Hence, 
\begin{align}\label{eq:stein_expectation}
\A_q f(x) 
&= \E_S \left[\E_{\{0,1\}}[f(X_{s},x_{-s}|x)]\right] - f(x)
. 
\end{align}
For a fixed {network}  $x$, we  seek a function $f \in \H$, s.t. $\|f\|_{\H}\leq 1$, that best distinguishes the difference in Eq.\eqref{eq:stein_expectation} {when $X$ does not have distribution $q$; this rationale} is similar {as for} Eq.\eqref{eq:ksd}. We
 define the {\it graph kernel Stein statistics} (gKSS) as
\begin{align}\label{eq:gkss}
    \operatorname{gKSS}(q;x) 
    & = \sup_{\|f\|_{\H}\leq 1} \Big|\E_S[\A^{(S)}_q f(x)] \Big|.
\end{align}
{It is often more convenient to consider $\operatorname{gKSS}^2(q;x) $. Let {the RKHS} $\H$
{have} kernel $K$ and inner product $\langle \cdot, \cdot \rangle_\H$.
By the reproducing property of RKHS functions, {as for Eq.\eqref{hform},} algebraic manipulation allows  the supremum to be computed in closed form:
\begin{align}\label{eq:gkss_quadratic_form}
{\operatorname{gKSS}}^2(q;x) =  \frac{1}{N^2} \sum_{s, s'\in [N]} h_x(s, s')
\end{align}
where 
$h_x(s, s') = \left\langle \A^{(s)}_q K(x,\cdot), \A^{(s')}_q K(\cdot,x)\right\rangle_{\H}.$
}

\paragraph{Stein Operator from Edge Re-sampling}
{When the distribution of $X$ is known, the} expectation in Eq.\eqref{eq:stein_expectation} can be computed for networks with a small number of vertices, but {when the number of vertices is large, exhaustive evaluation is computationally intensive.}
For a fixed network $x$, we propose the following randomised Stein operator via edge re-sampling. {This procedure mimics the {Markov process} which gives rise to the Stein operator.}  
Let $B$ be the {fixed} number of edges to be re-sampled. Our re-sampled Stein operator is
\begin{equation}\label{eq:resample_kss}
    \widehat{\A_q^B} f(x) = \frac{1}B \sum_{b\in [B]} \A^{(s_b)}_q f(x)
\end{equation}
where ${b \in B}$ and $s_b$ are edge samples from $\{1, \ldots, N\}$, {chosen uniformly with replacement, independent of each other and of $x$.}
The  expectation of  $ \widehat{\A_q^B} f(x)$ with respect to the re-sampling is 
$$
\E_B [ \widehat{\A_q^B} f(x)  ]
{ = \E_S [ \A_q^{(S)} f(x)]} 
= \A_q f(x). 
$$

We {introduce} the corresponding {re-sampling} gKSS:
\begin{align}\label{eq:gkss_resample}
    \widehat{\operatorname{gKSS}}(q;x) = \sup_{\|f\|_{\H}\leq 1} \Big|\frac{1}{B}\sum_{b\in [B]}\A^{(s_b)}_q f(x) \Big|.
\end{align}
{{The supremum in Eq.\eqref{eq:gkss_resample} is achieved by} $
f^*(\cdot) = {\frac{1}{B}\sum_b \A_q^{(s_b)}k(x,\cdot)} / {\left\| \frac{1}{B}\sum_a \A_q^{(s_a)}k(x,\cdot) \right\|}
$.} 
Similar algebraic manipulations as for Eq.\eqref{eq:gkss_quadratic_form} {yield} 

\begin{align}\label{eq:gkss_resample_quadratic_form}
    \widehat{\operatorname{gKSS}}^2(q;x) =  \frac{1}{B^2} \sum_{b, b'\in [B]} h_x(s_b, s_{b'}).
\end{align}

\section{GOODNESS-OF-FIT TEST with KERNEL STEIN STATISTICS}\label{sec:gof}
\subsection{Goodness-of-fit Testing Procedures}
{W}e now describe the proposed procedure to assess the goodness-of-fit of an ERGM for a single network observation.
The ERGM can be readily simulated from an unnormalised density via MCMC, {see for example}  \cite{morris2008ergm}. {Suppose that $q$ is the distribution of ERGM$(\beta, t)$ and $x$ is the observed network for which we want to assess the fit to $q$}. 
We simulate {independent} networks  $z_1,\dots,z_m \sim q$  and compare the {observed} $\widehat{\operatorname{gKSS}}^2(q;x)$ with 
the {set of} $\widehat{\operatorname{gKSS}}^2(q;z_i), {i=1, \ldots, m}$ {using a Monte Carlo test. As $\operatorname{gKSS}$ assesses the deviation from the null distribution, the test is one-sided; we reject the null model when the observed  $\widehat{\operatorname{gKSS}}$ is large}. 
The detailed test procedure is given in Algorithm \ref{alg:kernel_stein_monte_carlo}.

\begin{algorithm}[th]
   \caption{Kernel Stein Test for ERGM}
   \label{alg:kernel_stein_monte_carlo}
\begin{algorithmic}[1]
\renewcommand{\algorithmicrequire}{\textbf{Input:}}
\renewcommand{\algorithmicensure}{\textbf{Objective:}}
\REQUIRE~~\\
    Observed network $x$; \\
    Null model $q$; \\ 
    RKHS Kernel $K$; \\
    Re-sample size $B$; \\
    Confidence level $\alpha$;\\ 
    {Number of simulated networks} $m$;
\ENSURE~~\\
Test $H_0: x\sim q$ versus $H_1: x \not\sim q$.
\renewcommand{\algorithmicensure}{\textbf{Test procedure:}}
\ENSURE~~\\
\STATE Sample $\{s_1,\dots,s_B\}$ with replacement uniformly from $[N]$.
\STATE Compute $\tau =\widehat{\operatorname{gKSS}}^2(q;x)$ in Eq.\eqref{eq:gkss_resample_quadratic_form}.
\STATE Simulate $z_1,\dots,z_m \sim q$.
\STATE Compute $\tau_i =\widehat{\operatorname{gKSS}}^2(q;z_i)$ in Eq.\eqref{eq:gkss_resample_quadratic_form}. 
{again with re-sampling, choosing new samples 
 $\{s_{1, i},\dots,s_{B,i}\}$ uniformly from $[N]$} with replacement.
\STATE {Estimate} the {empirical}  $(1-\alpha)$-quantile $\gamma_{1-\alpha}$ of $\tau_1,\dots,\tau_m$.
\renewcommand{\algorithmicrequire}{\textbf{Output:}}
\REQUIRE~~\\
Reject $H_0$ if $\tau > \gamma_{1-\alpha}$; otherwise do not reject.
\end{algorithmic}
\end{algorithm}

\subsection{Kernel Choices}
\paragraph{Graph kernels}
Apart from using simple kernels between adjacency vectors {in} {$ \{0,1\}^N$}, we apply graph kernels that take into account graph topology via various measures. Various aspects of graph kernels have been studied  \citep{borgwardt2005shortest,vishwanathan2010graph,shervashidze2011weisfeiler,sugiyama2015halting}. We provide a brief review of some graph kernels in the Supplementary Material \ref{supp:graph_kernel}. In our implementation in R, we utilise the $\mathtt{ergm}$ package related to  \cite{morris2008ergm} for simulating ERGMs and the $\mathtt{graphkernels}$ package associated with \cite{sugiyama2018graphkernels} for computing relevant graph kernels.

\paragraph{Vector-valued RKHS}
As the operator in Eq.\eqref{eq:stein_expectation} has embedded a notion of conditional probability, we 
may {tailor} 
the RKHS {accordingly}. To incorporate the notion of $x_{s}$ conditioning on $x_{-s}$, we consider a  separate treatment of $x_s$ and $x_{-s}$ and introduce a vector-valued reproducing kernel Hilbert Space (vvRKHS). Similar constructions are  studied in \cite{jitkrittum2020testing} when testing goodness-of-fit for conditional densities. In the Supplementary Material \ref{supp:vvRKHS}, we provide a review on the vvRKHS we use {as}  graph kernels;
{further details
can be found in} \cite{caponnetto2008universal}, \cite{carmeli2010vector}, or \cite{sriperumbudur2011universality}.

\subsection{Theoretical Properties of gKSS}

{Let $X \sim q$ and  $Y \sim {\tilde{q}}$, where $ \tilde{q}$ is the distribution of an appropriately chosen ER graph.} Our theoretical approximation argument has three steps:  The first step, Theorem \ref{thm:asy_null_first}, is to approximate gKSS$(q, X)$ by  gKSS$({\tilde{q}}, Y)$, {with an explicit bound on the approximation error, as the number of vertices $n \rightarrow \infty$. 
Secondly, Theorem \ref{normalapprox} provides a normal approximation for gKSS$({\tilde{q}}, Y)^2$ of the approximating Bernoulli random graph {as $n \rightarrow \infty$}, again with an explicit bound, so that
{approximate confidence bounds for the test under the null hypothesis can be obtained.}
Finally, {a normal approximation for} 
$\widehat{\operatorname{gKSS}}(q, X)^2$ {to a normal distribution with approximate mean  gKSS$(q, X)$, as $B \rightarrow \infty$ with $\lfloor B/N \rfloor$ fixed, is
given in Proposition \ref{bootstrapnormal}, again with an explicit error bound. These three results combined provide explicit control} 
of the type 1 error.} 

{In \cite{chatterjee2013estimating} and \cite{reinert2019approximating} it is shown that under some conditions on the parameters, an ERGM$(\beta,t)$ is close to an approxpriately chosen Bernoulli random graph, as follows. }
{For $a\in[0,1]$, define the following}
functions \citep{bhamidi2011mixing,eldan2018exponential}, 
{with the notation in Definition~\ref{def:ergm}  for ERGM$(\beta, t)$}:
$$\Phi(a) := \sum_{l=1}^{k} \beta_l e_l a^{e_l -1}, \quad
\varphi(a) 
:= \frac{1+ \tanh(\Phi(a))}{2} 
$$
where $e_l$ is the number of edges in $H_l$. {For a polynomial $f(x) = \sum_{\ell=0}^k c_\ell x^\ell$ set 
$| f(x) | :=\sum_{\ell=1}^k |c_\ell| \, e_\ell \, x^{\ell}$. Moreover, $|| f ||$ denotes the supremum norm.}
The class of  ERGM$(\beta,t)$ {in this section are assumed to}   satisfy the following {standard} technical assumption.
\begin{Assumption}\label{assum:er_approx}
$\operatorname{(1)}$ $\frac{1}{2}|\Phi|'(1) < 1$.
$\operatorname{(2)}$ $\exists a^{*} \in [0,1]$ that solves the equation $\varphi(a^{*}) = a^{*}$.
\end{Assumption}
{The value $a^*$ will be the edge probability in the approximating Bernoulli random graph, $\operatorname{ER}(a^*)$. Then the following result holds.} 
\begin{Proposition}\label{prop:er_approx}\textnormal{[{Theorem 1.7 and} Corollary 1.10 \citep{reinert2019approximating}]}
Let  {$\operatorname{ERGM}(\beta,t)$} satisfy Assumption \ref{assum:er_approx}. Let $X \sim \operatorname{ERGM}(\beta)$ and $Y \sim \operatorname{ER}(a^*)$.
Then for $h:\{0,1\}^N\to \R$,
$$
|\E h(X)-\E h(Y) |
	\leq || {\Delta h} || N \frac{C_{a^*}(\beta, t, h) }{\sqrt{n}} \sum_{\ell=2}^k \beta_\ell.
$$
Here $C_{a^*}(\beta, t, h)$ is an explicit constant. 
\end{Proposition}
Proposition \ref{prop:er_approx} shows, {that for large $n$,}  the ERGM can be approximated {well} by an appropriate ER graph  {for} test functions $h$ which {are} properly scaled. {In particular, if $H$ is a connected graph and}  
$h(x) = t(H,x)n^{-|V(H)|}$, 
 then there is  {an explicit}  constant $C=C(\beta, t, H)$ such that
$
\left|\E h(X) - \E h(Y) \right| \leq C / \sqrt{n}.
$
{This result translates into an approximation for the gKSS, as follows.}

{\begin{Theorem}\label{thm:asy_null_first}
Let $q(x)=\operatorname{ERGM}(\beta, t)$ {satisfy} Assumption \ref{assum:er_approx} and let ${\tilde q}$ denote the distribution of {ER$(a^*)$.}
For $f\in \H$ equipped with kernel $K$, let
$
f_x^*(\cdot) = \frac{ (\A_q - \A_{\tilde q} ) K(x,\cdot)}{\left\|(\A_q - \A_{\tilde q} ) K(x,\cdot) \right\|_{\H}}.
$
Then there is an explicit constant $C=C(\beta, t, {K})$  such that for all $\epsilon > 0,$
\begin{eqnarray*}
\lefteqn{
\P ( |  \gKSS (q,X)  - \gKSS ( {\tilde q}, Y) | \, >  \,  \epsilon)}\\
&\le& \Big\{  || \Delta (\gKSS(q, \cdot))^2 || ( 1 +  || \Delta \gKSS(q, \cdot) || )   \\
&&  +  4 \sup_x { (|| \Delta f_x^*||^2) }  \Big\} 
{n \choose 2}  \frac{C}{{\epsilon^2 \sqrt{n}}} . 
\end{eqnarray*}
\end{Theorem}
}
{As our goodness-of-fit test statistic is based on the square of the $\operatorname{gKSS}$,  the asymptotic behaviour of $
 \operatorname{gKSS}^2 ( {\tilde q}, Y)$ is of interest. 
 To approximate the distribution of $\operatorname{gKSS}^2$ under the null hypothesis we make some additional assumptions on kernel $K$ of RKHS.
 
 
\begin{Assumption}\label{assum:er_approx_kernel}
Let $\H$ be the RKHS 
associated
with the  kernel $K: { \{ 0, 1 \}^N} \times  { \{ 0, 1 \}^N}\to \R$ and  for $s\in [N]$ let  $\H_s$ be the RKHS 
associated with the kernel $l_s: \{ 0, 1 \} \times  \{ 0, 1 \} \to \R$. Then \vspace{-2mm} 
\begin{enumerate}[i)]
    \item \label{ass21} $\H$ is a tensor product RKHS, $\H = \otimes_{s \in [n]}\H_s $; \vspace{-2mm} 
    \item \label{ass22} $k$ is a product kernel, $k(x, y) = \otimes_{s \in [N]} l_s(x_s, y_s)$; \vspace{-2mm} 
    \item \label{ass23} $ \langle l_s (x_s, \cdot), l_s (x_s, \cdot) \rangle_{\H_s}  =1$; \vspace{-2mm} 
    \item \label{ass24} $l_s(1, \cdot) - l_s(0, \cdot) \ne 0$ for all $s \in [N]$. \vspace{-2mm} 
\end{enumerate}
\end{Assumption}
These assumptions are satisfied for example for the suitably standardised Gaussian kernel $K(x,y) = \exp\{ - \frac{1}{\sigma^2} \sum_{s \in [N]} (x_s - y_s)^2\} $.

Letting 
$|| \cdot ||_1$ denote  $L_1$-distance, and $\mathcal L$ denote the law of a random variable, {in Supplementary Material \ref{supp:proofs}}, we show the following normal approximation.}

{ \begin{Theorem} \label{normalapprox}
 Assume that the conditions i) - iv) in Assumption \ref{assum:er_approx_kernel} hold. 
 Let $\mu = \E [  \operatorname{gKSS}^2 ( {\tilde q}, Y)] $ and $\sigma^2 = \Var [ \operatorname{gKSS}^2 ( {\tilde q}, Y)]. $ Set 
 $W = \frac{1}{\sigma} (  \operatorname{gKSS}^2 ( {\tilde q}, Y)]  - \mu)$ and let  $Z$ denote a standard normal variable, Then there is an explicit constant $C = C(a^*, l_s, s \in [N]) $ such that 
\begin{equation*} 
    || {\mathcal L}(W) - {\mathcal L}(Z)||_1  \le \frac{C}{\sqrt{N}}. 
    \end{equation*} 
 \end{Theorem}
 More details on $\mu$ and $\sigma^2$ are given in the Supplementary Material {\ref{supp:proofs}}. 
 This normal approximation could also be used to assess the asymptotic distribution under an  alternative $x \sim p$ where $p(x)=\operatorname{ERGM}(\beta',t')$ satisfies Assumption 1 with edge probability $b^*$ and $b^* \neq a^*$. Then asymptotically we can compare the corresponding  normal random variables with different means. 
 }

 For the final step, the re-sampling version, let ${k_s}$ be the number of times that $s$ is included in the sample ${\mathcal B}$, {where $|{\mathcal B}| = B $}.
 Then, from 
Eq.\eqref{eq:gkss_resample_quadratic_form}, 
$$  \widehat{\operatorname{gKSS}}^2(q;x)
  =  \frac{1}{B^2} \sum_{s, s' \in [N] } k_s k_{s'}  h_x(s, s').
$$
In this expression, the randomness  only lies in the counts  {$k_s$, where $s \in [N]$}. These counts are exchangeable {and 
${\mathbf{k}} = (k_s, s  \in [N])$ follows the multinomial
   $(B; N^{-1}, \ldots, N^{-1})$ distribution.} Hence the statistic 
$  \frac{1}{N^2} \sum_{s,t  \in [N] } k_s k_t  h_x(s, t)$ is a sum of weakly globally dependent random variables, {although due to the network $x$ being fixed, this is not a classical $V$-statistic}. {Instead, Stein's method will be used to prove the following result in the Supplementary Material.} 

\begin{Proposition}\label{bootstrapnormal} 
Let    
$$Y  = \frac{1}{B^2}  \sum_{s, t \in [N] } ( k_s k_t - \E (k_s k_t) ) h_x (s, t).$$
Assume that  $h_x$ is  bounded such that $Var(Y)$ is non-zero. 
Then if $Z$ is mean zero normal with variance $Var(Y)$, there is an {explicitly computable}  constant $C>0$ such that for all three times  continuously differentiable functions $g$ with bounded derivatives up to order 3, 
$$
| \E [ g(Y) ] - \E [g(Z)] \le \frac{C}{B}. 
$$ 
\end{Proposition} 
{When {the sampling fraction} $F= \frac{B}{N}$ is kept approximately constant as $N\rightarrow \infty$,} 
noting that
\begin{eqnarray*}
\widehat{\operatorname{gKSS}}^2(q;x)  
   &=&   Y + \operatorname{gKSS}^2 + \frac{N-1}{B N^2} \sum_{s \in [N]} h (s, s) \\
   && - \frac{1}{ B^2 N^2}\sum_{s, t \in [N] , s \ne t}  h (s, t) 
   \end{eqnarray*}  
   the  normal approximation {for $\widehat{\operatorname{gKSS}}^2(q;x)$ with approximate mean  $\operatorname{gKSS}^2(q;x)$}    follows for $N \rightarrow \infty$.
}

\section{EXPERIMENTS}\label{sec:exp}
{To assess the performance of the test},
we replicate the {synthetic} {benchmark-type} settings from \citep{lusher2013exponential,rolls2015simulation,yang2018goodness}. 
We then apply our tests to {three} real data {networks: Lazega's lawyer network \citep{lazega2001collegial} and a friendship network \citep{steglich2006applying} which are both} studied in \citep{yin2019selection},  as well as {a co-sponsorship network from}  \citep{fowler2006connecting, fowler2006legislative}.

\subsection{Synthetic Experiments}

\begin{figure*}[ht!]
	\begin{center}
		\subfigure[$n=20$, $\alpha=0.05$,  $\beta_2$ perturbed]
		{\includegraphics[width=0.48\textwidth, height=0.32\textwidth]{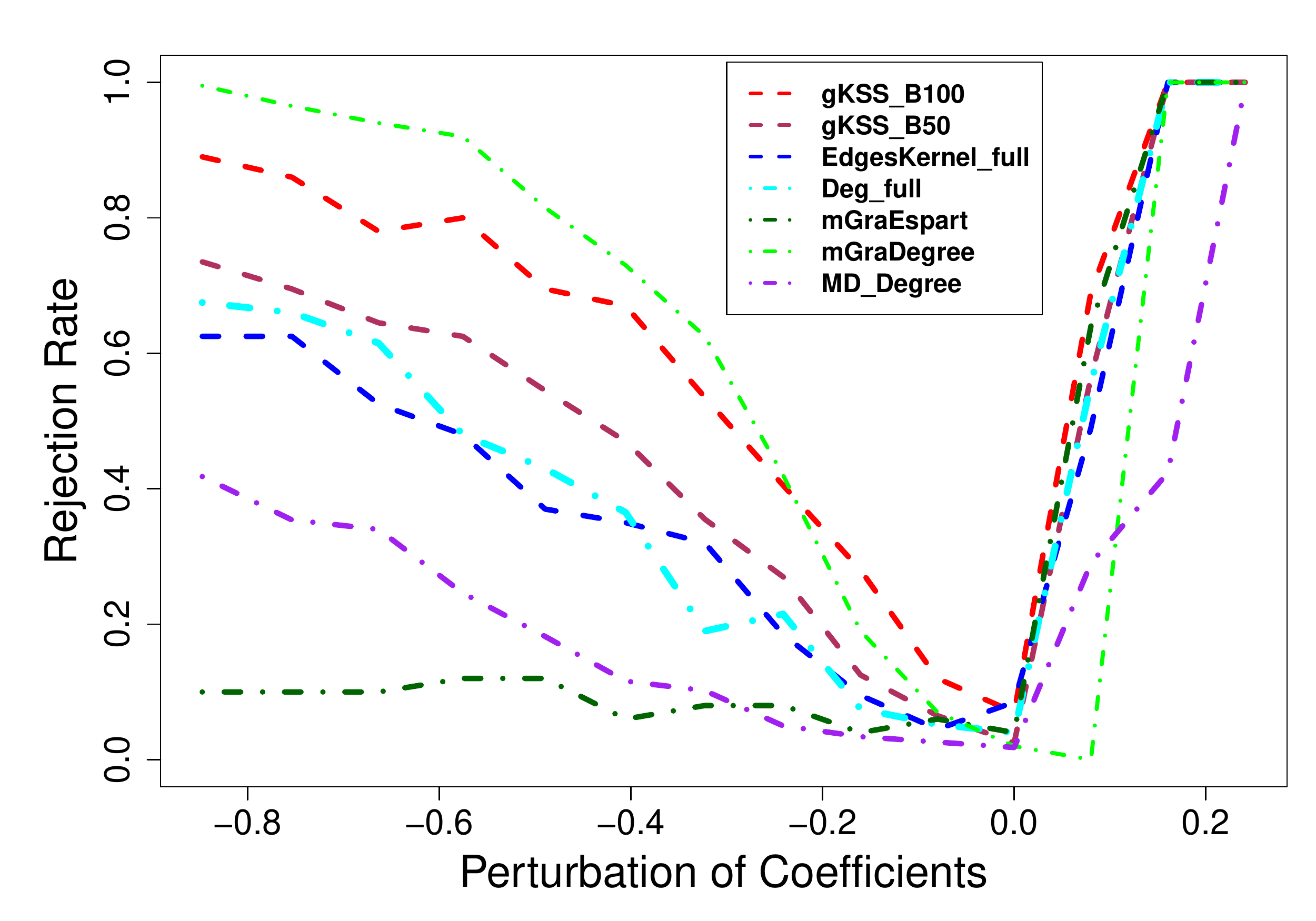}\label{fig:perturb}}
		\subfigure[$H_1: \beta = (-2,-0.03, 0.01)$, $\alpha=0.05$]
		{\includegraphics[width=0.48\textwidth, height=0.32\textwidth]{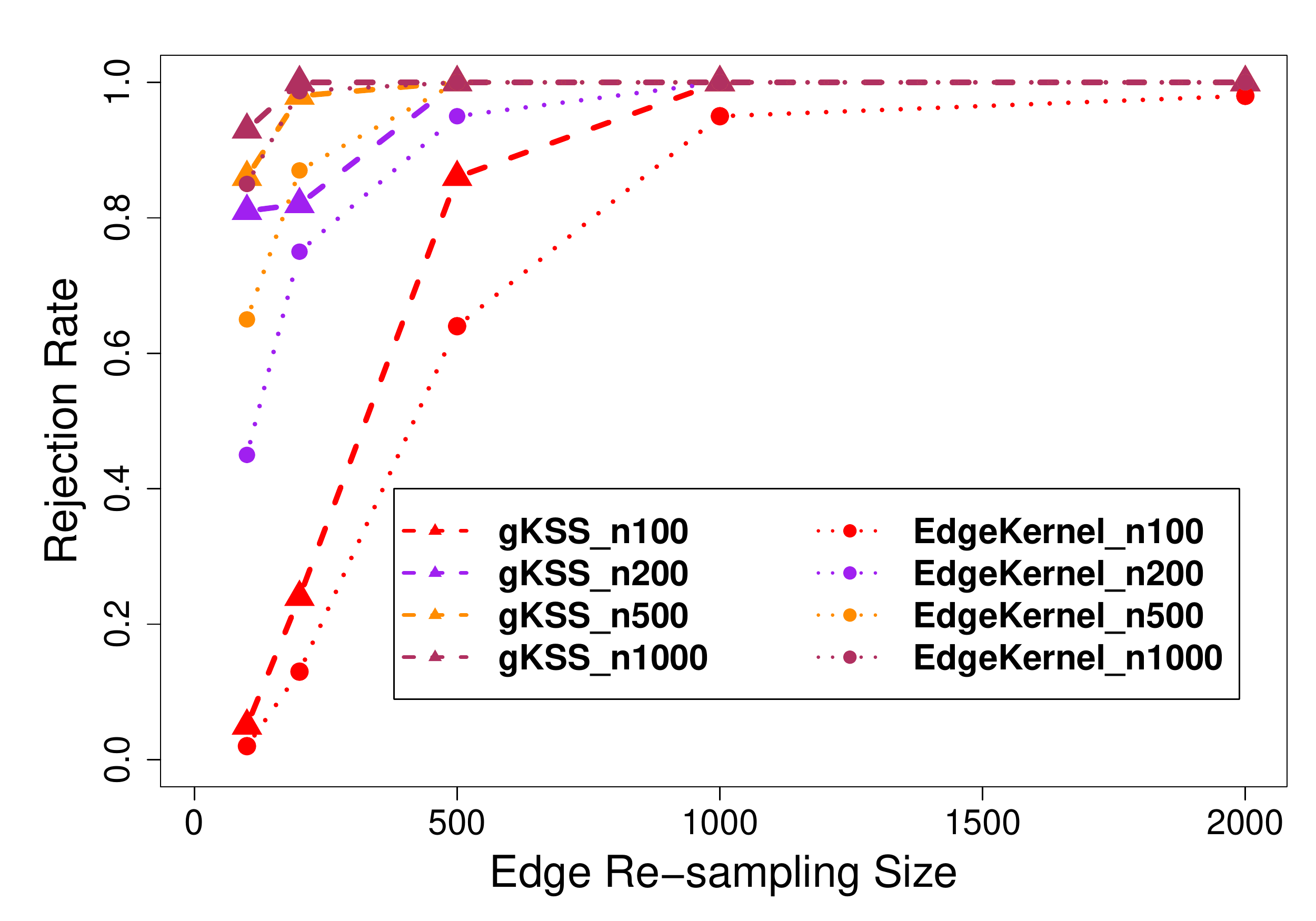}\label{fig:large_network}}
		\caption{Simulation Results for E2ST Model}	\label{fig:synthetic-problems}	
	\end{center}
\end{figure*}

{\bf Model }
In the synthetic example, we assess the test performance on relatively simple but useful ERGMs, with {three graphs $H_l$ {in the statistic $t$}, namely} edge, 2-star, and triangle; {we abbreviate this model as E2ST}.
Then the unnormalised density has the form
\begin{equation}\label{eq:e2st}
q(x) \propto \exp{\Big( \beta_1 E_d(x) + \beta_2 S_2(x)+ \beta_3 T_r(x) \Big)},
\end{equation}
where $E_d(x)$ denotes the number of edges in $x$; $S_2(x)$ denotes the number of 2-stars in $x$ and $T_r(x)$ denotes the number of triangles in $x$.
We choose the null parameter as $(\beta_1,\beta_2,\beta_3) = (-2, 0.0, 0.01)$, which satisfies Assumption \ref{assum:er_approx} {and gives} 
$a^*= 0.1176$. {For the alternative distributions, 
following similar settings in \citep{yang2018goodness}, we fix the coefficient $\beta_1= -2$ and $\beta_3 =0.01$ of the E2ST model in Eq.\eqref{eq:e2st}, and test the null model of $H_0: \beta_2 = 0$ against the alternative $H_1: \beta_2 \neq 0$  with a perturbed $\beta_2$ so that the alternative model satisfies Assumption 1 also.
}


{\bf The Proposed Methods} 
We apply the proposed goodness-of-fit test procedures and compare with existing approaches. {We use the following abbreviations:}  
\textbf{gKSS} stands for the proposed test in Algorithm \ref{alg:kernel_stein_monte_carlo}; 
\textbf{gKSS\_B100} uses Eq.\eqref{eq:gkss_resample_quadratic_form} as the test statistic where $B=100$, and  \textbf{gKSS\_n20} denotes testing the problem with  $n=20$ {vertices}. Results shown in Fig.~\ref{fig:synthetic-problems} are based on  Weisfeiler-Lehman graph kernels \citep{shervashidze2011weisfeiler}; {results using other kernels are shown
in the Supplementary Material Section \ref{supp:exp}}. 
\textbf{EdgeKernel} denotes the gKSS with a kernel between binary edges which corresponds to  a test based on edge counts. Re-sampling applies, e.g. \textbf{EdgeKernel\_B100} indicates that $100$ edges are re-sampled from the network.

{\bf The Competing Approaches} We list the goodness-of-fit testing methods {which serve as comparisons using the following} abbreviations. 
\textbf{Degree\_full} stands for degree-based tests \citep{ouadah2020degree}, 
where the variance of degree counts on vertices are used as test statistics. 
The suffix ``full'' indicates that all vertices are used. 
The graphical tests for goodness-of-fit {from}  \cite{hunter2008goodness} simulate the null distribution of a chosen network statistic from the null model as a visual guideline 
for goodness-of-fit. We {quantify this}  idea by using total variation (TV) distance between distributions of network statistics of choice as test statistics; \textbf{mGra} stands for the modified graphical test, where the 
{TV} 
distance is  used to compare the distribution of the summary statistics of choice. Full details are provided in the Supplementary Material \ref{supp:mgra_tv}.
We append  \textbf{mGra} {by the  summary statistics used}, so that, {for example},
\textbf{mGraDegree} uses the TV distance between degree distributions as test statistics. \textbf{Espart} (or espartners) stands for edgewise shared partner \citep{hunter2008goodness}; 
\textbf{MD\_Degree} stands for the test based on Mahalanobis distance between chosen summary statistics \citep{lospinoso2019goodness}. The suffix after hyphen indicates that the vertex degree is used as network statistics.

{{\bf Test Results}}
The {main} results are shown  in Fig.~\ref{fig:perturb}. We see that $\operatorname{gKSS}$ has higher power than the competitors, while, as expected, larger re-sampling size performs better.
{A} denser networks can be easier to distinguish {as higher subgraph counts are available compared to}
sparser networks. {In our experimental set-up,}
the network size $n=20$ is relatively small and the null model, with $\beta_1 = -2$, is {fairly} sparse. 
{We} observe that \textbf{mGraDegree} has slightly higher power than $\operatorname{gKSS}$ when  $\beta_2 < -0.3$ {so that} the graph is sparser; 
it performs poorly when the alternative model is closer to the null, i.e. $|\beta_2|$ small \footnote{In particular, it did not identify the slightly denser alternatives, which should be relatively easier problems.}. This 
{may  relate to using} the TV distance  for {comparing the}  degree distribution; {the phenomenon does not occur for \textbf{MD\_Degree}}. Overall, \textbf{gKSS} is more reliable and has {typically} higher power {compared to}
{these competing}  methods.

{\bf Increasing Edge Re-sampling Size B} 
Fig.\ref{fig:large_network} shows the test power of large networks up to $n=1000$ {vertices}. The results show that the tests achieve maximal power with a relatively small number of re-sampling edges {indicators}.
With the choice of re-sampling size $B$ and good test power with a relatively small number of re-sampled edge indicators, $\operatorname{gKSS}$ is applicable to 
{networks with a large number of vertices, beyond the reach of the} 
graphical-based tests \citep{hunter2008goodness}.
In particular, the proposed tests can be useful in validating model assumptions in practical problems where the networks
{have a large number of vertices}.

{\bf Computational Time} {The computational times for each test are shown in Table~\ref{tab:runtime}.} {The {\textbf{gKSS}} tests are faster than the {\textbf{mGra}} tests and of similar speed as the less accurate full degree method.}
{The slow} {\textbf{mGra}}
tests are based on the {computational demanding as well as hard-to-scale} estimation associate with the graphical-based method  in \cite{hunter2008goodness}.
{Its}  main computational cost {stems} from 
simulating the null graphs from $\mathtt{ergm}$ to compute the TV distances.
{Although} the  \textbf{Degree\_full} {test} is supposed to be fast with computational complexity $O(n)$ , due to the estimation of the mean and variance of the degree statistics via simulating the null from $\mathtt{ergm}$, its runtime is comparable with \textbf{gKSS\_B50} {with}  complexity  $O(B^2)$ for $B=100$.

\begin{table}[ht]
\centering
\begin{tabular}{crrr}
  \hline
n & gKSS\_B50 & gKSS\_B100 & gKSS\_B200 \\ 
  \hline
20 & 14.53 & 33.57 & 67.63 \\ 
  30 & 15.03 & 41.08 & 70.14 \\ 
  50 & 21.54 & 50.10 & 91.18 \\ 
   \hline
  \hline
n & Degree\_full & mGraDegree & mGraEspart \\ 
  \hline
20 & 38.08 & 4596.67 & 4779.04 \\ 
  30 & 39.08 & 4840.66 & 4871.72 \\ 
  50 & 44.09 & 5127.74 & 5210.40 \\ 
   \hline
\end{tabular}
    \caption{The computational time for each test, in seconds, for 500 trials.}
    \label{tab:runtime}
\end{table}


\subsection{Real Data Applications}
Next we apply our test to two {benchmark {social} network data sets} {which are analysed in \cite{yin2019selection};} 
Lazega’s lawyer network  \citep{lazega2001collegial} consists of a 
network between 36 laywers;  the Teenager friendship {network} \citep{steglich2006applying} {is a friendship data set of}  50 secondary school students in Glasglow.
Moreover, we  apply our proposed test to large  network, a co-sponsorship network for pieces of  legislation in the U.S. Senate  from  \cite{fowler2006connecting,fowler2006legislative}. The network {data used here are from} 
\cite{schmid2017exponential} {and}  consists of $2825$ vertices and 
$28813$ edges. {For all three networks we fit an ER model with the maximum likelihood estimate as  edge probability,  an E2ST model, and  an ER($a^*$) model using as edge probability $a^*$ calculated from the E2ST fit, or, for the co-sponsorship network, calculated from fitting an additional model detailed below. Table~\ref{tab:real_data} summarises the results.} 

For the Lawyer network, \cite{lazega2001collegial} suggests
{an ER model}.
Our test does not reject this null hypothesis when testing against the best fitted ER graph, {with edge probability $p= 0.055$}, 
which supports the assumed model. The fitted E2ST model with $\beta=(-2.8547,  -0.0003,   0.6882)$ is  {rejected} at $\alpha=0.05$. {This E2ST is close to an ER graph with  $\beta_{E_d} = -2.774$ and the corresponding ER($a^*$) model is not rejected at $\alpha=0.05$.}  

For the Teenager network,  
the fitted ER model with $p=0.046$
is rejected at $\alpha=0.05$; 
{for the} fitted E2ST model in Eq.\eqref{eq:e2st} with $\beta=(-2.3029,   -0.3445,    2.8240  )$
we do not have strong evidence to reject the null at $\alpha=0.05$. {The corresponding ER($a^*$) model is also not rejected at $\alpha=0.05$.}  {In particular the maximum likelihood estimator does not give the best fitting ER model.} 


The co-sponsorship network is well fitted by the ER graph with edge probability $p=0.0072$.
{In contrast,} the fitted E2ST with $\beta=(-6.4126,  -0.0240,   2.4684)$,  is rejected at $\alpha=0.05$. Additionally 
we  fit the ERGM proposed in \cite{schmid2017exponential}, which includes party homophily \citep{zhang2008community} and {the alternating $k$-star statistic} \citep{snijders2006new}:
$$
{q^*}(x) \propto \exp{\{\beta_1 E_d(x) + \beta_2\Gamma(x;P)} + \beta_3 {\rm S_{alt}}(x;\lambda)\}
$$
where $P$ denotes the party assignment information between the {pieces of} legislations, and $\Gamma(x;P)=\sum_{ij}x_{ij}P_{ij}$; {with} the $k$-star count $S_k(x)$, the alternating $k$-star statistic is  ${\rm S_{alt}}(x;\lambda)= \sum_{k=2}^{n-1}(-\frac{1}{\lambda})^{k-2} S_k(x) 
$.
{We use the model $q^*$} with parameters fitted in \cite{schmid2017exponential},  $\beta_1=-5.884$,  $\beta_2=1.440$, $\beta_3=0.124$, and the parameter in alternating $k$-stars $\lambda = 0.4975$. {This model (with p-value=$0.022$), as well as its corresponding ER($a^*$) model}  are rejected at $\alpha = 0.05$. 


\begin{table}[]
    \centering
\begin{tabular}{c|c|c|c|c}
\hline
{} & $n$&  ER &   E2ST  &  ER($a^*$) \\
\hline
Lawyer & 36&      \color{blue}{0.280} &   \color{red}{0.012}   & \color{blue}{0.152}     \\
Teenager & 50  & \color{red}{0.016}      & \color{blue}{0.060}  &   \color{blue}{0.336} \\
Co-sponsor & 2825& \color{blue}{0.612} & \color{red}{0.002}  &   \color{red}{0.036} \\
\hline
\end{tabular}
\vspace{0.3cm}
    \caption{Rejection rates for real networks. The results marked ${\color{blue}{blue}}$ indicate {not rejecting} and 
    ${\color{red}{red}}$ 
    the null hypothesis at $\alpha=0.05$, using {$\widehat{\operatorname{gKSS}^2}$ with} $B=200$.}
    \label{tab:real_data}
\end{table}

\section{CONCLUSIONS AND DIRECTIONS FOR FURTHER WORK}
{In this paper we {provide} a novel goodness-of-fit test for exponential random graph models using Stein's method. A key feature is that the test relies on the observation of only one network. Probabilistic properties of the test statistic are analysed through comparison with Bernoulli random graphs. 
}

{Directions for future work include a thorough analysis of the interplay of the graph kernels used in the RKHS and the GKSD. Adaptive methods for tuning graph kernel hyper-parameters would be interesting; see for example \cite{gretton2012optimal} or \cite{jitkrittum2016adaptive}.}

{Further, a large contribution to the computational cost of GKSD stems from sampling from the null model; an issue which affects all main methods for assessing goodness-of-fit for exponential random graph models. Developing a goodness-of-fit testing procedure based on a single network observation which does not require simulations from the null model is an exciting future challenge. 
}

{Finally, the approach is of independent interest and holds promise for adaptation to other random graph models.}

\paragraph{Acknowledgements}
The authors would like to thank Arthur Gretton and Andrew Barbour for helpful discussions. {Moreover, they thank the anonymous reviewers for many good comments which have improved the paper.}  G.R. acknowledges the support from EP/R018472/1 as well as from the COSTNET COST Action CA 15109. 
W.X. acknowledges the support from the Gatsby Charitable Foundation.

\bibliographystyle{abbrvnat}
\bibliography{stein}

\begin{thebibliography}{67}
\providecommand{\natexlab}[1]{#1}
\providecommand{\url}[1]{\texttt{#1}}
\expandafter\ifx\csname urlstyle\endcsname\relax
  \providecommand{\doi}[1]{doi: #1}\else
  \providecommand{\doi}{doi: \begingroup \urlstyle{rm}\Url}\fi

\bibitem[Barbour and Chen(2005)]{barbour2005introduction}
A.~Barbour and L.~Chen.
\newblock An introduction to {S}tein‘s method.
\newblock \emph{Lecture Notes Series. Institute for Mathematical Sciences.
  National University of Singapore}, 4, 2005.

\bibitem[Barp et~al.(2019)Barp, Briol, Duncan, Girolami, and
  Mackey]{barp2019minimum}
A.~Barp, F.-X. Briol, A.~Duncan, M.~Girolami, and L.~Mackey.
\newblock Minimum stein discrepancy estimators.
\newblock In \emph{Advances in Neural Information Processing Systems}, pages
  12964--12976, 2019.

\bibitem[Berlinet and Thomas(2004)]{RKHSbook}
A.~Berlinet and C.~Thomas.
\newblock \emph{Reproducing {K}ernel {H}ilbert {S}paces in {P}robability and
  {S}tatistics}.
\newblock Kluwer Academic Publishers, 2004.

\bibitem[Bhamidi et~al.(2011)Bhamidi, Bresler, and Sly]{bhamidi2011mixing}
S.~Bhamidi, G.~Bresler, and A.~Sly.
\newblock Mixing time of exponential random graphs.
\newblock \emph{The Annals of Applied Probability}, 21\penalty0 (6):\penalty0
  2146--2170, 2011.

\bibitem[Bhattacharyya and Bickel(2015)]{bhattacharyya2015subsampling}
S.~Bhattacharyya and P.~J. Bickel.
\newblock Subsampling bootstrap of count features of networks.
\newblock \emph{The Annals of Statistics}, 43\penalty0 (6):\penalty0
  2384--2411, 2015.

\bibitem[Bonferroni(1936)]{bonferroni1936teoria}
C.~Bonferroni.
\newblock Teoria statistica delle classi e calcolo delle probabilita.
\newblock \emph{Pubblicazioni del R Istituto Superiore di Scienze Economiche e
  Commericiali di Firenze}, 8:\penalty0 3--62, 1936.

\bibitem[Borgwardt and Kriegel(2005)]{borgwardt2005shortest}
K.~M. Borgwardt and H.-P. Kriegel.
\newblock Shortest-path kernels on graphs.
\newblock In \emph{Fifth IEEE International Conference on Data Mining
  (ICDM'05)}, pages 8--pp. IEEE, 2005.

\bibitem[Bresler and Nagaraj(2019)]{bresler2019stein}
G.~Bresler and D.~Nagaraj.
\newblock Stein’s method for stationary distributions of {M}arkov chains and
  application to {I}sing models.
\newblock \emph{The Annals of Applied Probability}, 29\penalty0 (5):\penalty0
  3230--3265, 2019.

\bibitem[Caponnetto et~al.(2008)Caponnetto, Micchelli, Pontil, and
  Ying]{caponnetto2008universal}
A.~Caponnetto, C.~A. Micchelli, M.~Pontil, and Y.~Ying.
\newblock Universal multi-task kernels.
\newblock \emph{Journal of Machine Learning Research}, 9\penalty0
  (Jul):\penalty0 1615--1646, 2008.

\bibitem[Carmeli et~al.(2010)Carmeli, De~Vito, Toigo, and
  Umanit{\'a}]{carmeli2010vector}
C.~Carmeli, E.~De~Vito, A.~Toigo, and V.~Umanit{\'a}.
\newblock Vector valued reproducing kernel {H}ilbert spaces and universality.
\newblock \emph{Analysis and Applications}, 8\penalty0 (01):\penalty0 19--61,
  2010.

\bibitem[Chatterjee and Diaconis(2013)]{chatterjee2013estimating}
S.~Chatterjee and P.~Diaconis.
\newblock Estimating and understanding exponential random graph models.
\newblock \emph{The Annals of Statistics}, 41\penalty0 (5):\penalty0
  2428--2461, 2013.

\bibitem[Chatterjee and Meckes(2008)]{chatterjee2008multivariate}
S.~Chatterjee and E.~Meckes.
\newblock Multivariate normal approximation using exchangeable pairs.
\newblock \emph{Alea}, 4:\penalty0 257--283, 2008.

\bibitem[Chen et~al.(2010)Chen, Goldstein, and Shao]{chen2010}
L.~H.~Y. Chen, L.~Goldstein, and Q.~M. Shao.
\newblock \emph{Normal approximation by {S}tein's method}.
\newblock Springer, 2010.

\bibitem[Chen and Onnela(2019)]{chen2019bootstrap}
S.~Chen and J.-P. Onnela.
\newblock A bootstrap method for goodness of fit and model selection with a
  single observed network.
\newblock \emph{Scientific Reports}, 9\penalty0 (1):\penalty0 1--12, 2019.

\bibitem[Chwialkowski et~al.(2016)Chwialkowski, Strathmann, and
  Gretton]{chwialkowski2016kernel}
K.~Chwialkowski, H.~Strathmann, and A.~Gretton.
\newblock A kernel test of goodness of fit.
\newblock In \emph{JMLR: Workshop and Conference Proceedings}, 2016.

\bibitem[Chwialkowski et~al.(2014)Chwialkowski, Sejdinovic, and
  Gretton]{chwialkowski2014wild}
K.~P. Chwialkowski, D.~Sejdinovic, and A.~Gretton.
\newblock A wild bootstrap for degenerate kernel tests.
\newblock In \emph{Advances in Neural Information Processing Systems}, pages
  3608--3616, 2014.

\bibitem[Eldan and Gross(2018)]{eldan2018exponential}
R.~Eldan and R.~Gross.
\newblock Exponential random graphs behave like mixtures of stochastic block
  models.
\newblock \emph{The Annals of Applied Probability}, 28\penalty0 (6):\penalty0
  3698--3735, 2018.

\bibitem[Fernandez et~al.(2020)Fernandez, Rivera, Xu, and
  Gretton]{tamara2020kernelized}
T.~Fernandez, N.~Rivera, W.~Xu, and A.~Gretton.
\newblock Kernelized stein discrepancy tests of goodness-of-fit for
  time-to-event data.
\newblock In \emph{International Conference on Machine Learning}, pages
  3112--3122. PMLR, 2020.

\bibitem[Fowler(2006{\natexlab{a}})]{fowler2006connecting}
J.~H. Fowler.
\newblock Connecting the {C}ongress: A study of cosponsorship networks.
\newblock \emph{Political Analysis}, 14\penalty0 (4):\penalty0 456--487,
  2006{\natexlab{a}}.

\bibitem[Fowler(2006{\natexlab{b}})]{fowler2006legislative}
J.~H. Fowler.
\newblock Legislative cosponsorship networks in the {US} {H}ouse and {S}enate.
\newblock \emph{Social Networks}, 28\penalty0 (4):\penalty0 454--465,
  2006{\natexlab{b}}.

\bibitem[Frank and Strauss(1986)]{frank1986markov}
O.~Frank and D.~Strauss.
\newblock Markov graphs.
\newblock \emph{Journal of the American Statistical Association}, 81\penalty0
  (395):\penalty0 832--842, 1986.

\bibitem[G{\"a}rtner et~al.(2003)G{\"a}rtner, Flach, and
  Wrobel]{gartner2003graph}
T.~G{\"a}rtner, P.~Flach, and S.~Wrobel.
\newblock On graph kernels: Hardness results and efficient alternatives.
\newblock In \emph{Learning Theory and Kernel Machines}, pages 129--143.
  Springer, 2003.

\bibitem[Gorham and Mackey(2015)]{gorham2015measuring}
J.~Gorham and L.~Mackey.
\newblock Measuring sample quality with {S}tein's method.
\newblock In \emph{Advances in Neural Information Processing Systems}, pages
  226--234, 2015.

\bibitem[Gretton et~al.(2012)Gretton, Sejdinovic, Strathmann, Balakrishnan,
  Pontil, Fukumizu, and Sriperumbudur]{gretton2012optimal}
A.~Gretton, D.~Sejdinovic, H.~Strathmann, S.~Balakrishnan, M.~Pontil,
  K.~Fukumizu, and B.~K. Sriperumbudur.
\newblock Optimal kernel choice for large-scale two-sample tests.
\newblock In \emph{Advances in Neural Information Processing Systems}, pages
  1205--1213, 2012.

\bibitem[Holland and Leinhardt(1981)]{holland1981exponential}
P.~W. Holland and S.~Leinhardt.
\newblock An exponential family of probability distributions for directed
  graphs.
\newblock \emph{Journal of the American Statistical Association}, 76\penalty0
  (373):\penalty0 33--50, 1981.

\bibitem[Hunter et~al.(2008)Hunter, Goodreau, and Handcock]{hunter2008goodness}
D.~R. Hunter, S.~M. Goodreau, and M.~S. Handcock.
\newblock Goodness of fit of social network models.
\newblock \emph{Journal of the American Statistical Association}, 103\penalty0
  (481):\penalty0 248--258, 2008.

\bibitem[Hyv{\"a}rinen(2005)]{hyvarinen2005estimation}
A.~Hyv{\"a}rinen.
\newblock Estimation of non-normalized statistical models by score matching.
\newblock \emph{Journal of Machine Learning Research}, 6\penalty0
  (Apr):\penalty0 695--709, 2005.

\bibitem[Jitkrittum et~al.(2017{\natexlab{a}})Jitkrittum, Szab{\'o}, and
  Gretton]{jitkrittum2016adaptive}
W.~Jitkrittum, Z.~Szab{\'o}, and A.~Gretton.
\newblock An adaptive test of independence with analytic kernel embeddings.
\newblock In \emph{International Conference on Machine Learning}, pages
  1742--1751. PMLR, 2017{\natexlab{a}}.

\bibitem[Jitkrittum et~al.(2017{\natexlab{b}})Jitkrittum, Xu, Szab{\'o},
  Fukumizu, and Gretton]{jitkrittum2017linear}
W.~Jitkrittum, W.~Xu, Z.~Szab{\'o}, K.~Fukumizu, and A.~Gretton.
\newblock A linear-time kernel goodness-of-fit test.
\newblock In \emph{Advances in Neural Information Processing Systems}, pages
  262--271, 2017{\natexlab{b}}.

\bibitem[Jitkrittum et~al.(2018)Jitkrittum, Kanagawa, Sangkloy, Hays,
  Sch{\"o}lkopf, and Gretton]{jitkrittum2018informative}
W.~Jitkrittum, H.~Kanagawa, P.~Sangkloy, J.~Hays, B.~Sch{\"o}lkopf, and
  A.~Gretton.
\newblock Informative features for model comparison.
\newblock In \emph{Advances in Neural Information Processing Systems}, pages
  808--819, 2018.

\bibitem[Jitkrittum et~al.(2020)Jitkrittum, Kanagawa, and
  Sch{\"o}lkopf]{jitkrittum2020testing}
W.~Jitkrittum, H.~Kanagawa, and B.~Sch{\"o}lkopf.
\newblock Testing goodness of fit of conditional density models with kernels.
\newblock In \emph{Conference on Uncertainty in Artificial Intelligence}, pages
  221--230. PMLR, 2020.

\bibitem[Kanagawa et~al.(2019)Kanagawa, Jitkrittum, Mackey, Fukumizu, and
  Gretton]{kanagawa2019kernel}
H.~Kanagawa, W.~Jitkrittum, L.~Mackey, K.~Fukumizu, and A.~Gretton.
\newblock A kernel {S}tein test for comparing latent variable models.
\newblock \emph{arXiv preprint arXiv:1907.00586}, 2019.

\bibitem[Kriege et~al.(2016)Kriege, Giscard, and Wilson]{kriege2016valid}
N.~M. Kriege, P.-L. Giscard, and R.~Wilson.
\newblock On valid optimal assignment kernels and applications to graph
  classification.
\newblock In \emph{Advances in Neural Information Processing Systems}, pages
  1623--1631, 2016.

\bibitem[Kriege et~al.(2020)Kriege, Johansson, and Morris]{kriege2020survey}
N.~M. Kriege, F.~D. Johansson, and C.~Morris.
\newblock A survey on graph kernels.
\newblock \emph{Applied Network Science}, 5\penalty0 (1):\penalty0 1--42, 2020.

\bibitem[Lazega(2001)]{lazega2001collegial}
E.~Lazega.
\newblock \emph{The collegial phenomenon: The social mechanisms of cooperation
  among peers in a corporate law partnership}.
\newblock Oxford University Press on Demand, 2001.

\bibitem[Ley et~al.(2017)Ley, Reinert, and Swan]{ley2017stein}
C.~Ley, G.~Reinert, and Y.~Swan.
\newblock Stein’s method for comparison of univariate distributions.
\newblock \emph{Probability Surveys}, 14:\penalty0 1--52, 2017.

\bibitem[Liu and Wang(2016)]{liu2016stein}
Q.~Liu and D.~Wang.
\newblock Stein variational gradient descent: A general purpose bayesian
  inference algorithm.
\newblock In \emph{Advances In Neural Information Processing Systems}, pages
  2378--2386, 2016.

\bibitem[Liu et~al.(2016)Liu, Lee, and Jordan]{liu2016kernelized}
Q.~Liu, J.~Lee, and M.~Jordan.
\newblock A kernelized {S}tein discrepancy for goodness-of-fit tests.
\newblock In \emph{International Conference on Machine Learning}, pages
  276--284, 2016.

\bibitem[Lospinoso and Snijders(2019)]{lospinoso2019goodness}
J.~Lospinoso and T.~A. Snijders.
\newblock Goodness of fit for stochastic actor-oriented models.
\newblock \emph{Methodological Innovations}, 12\penalty0 (3):\penalty0
  2059799119884282, 2019.

\bibitem[Lusher et~al.(2013)Lusher, Koskinen, and
  Robins]{lusher2013exponential}
D.~Lusher, J.~Koskinen, and G.~Robins.
\newblock \emph{Exponential {R}andom {G}raph {M}odels for {S}ocial {N}etworks:
  {T}heory, {M}ethods, and {A}pplications}.
\newblock Cambridge University Press, 2013.

\bibitem[Meckes(2009)]{meckes2009stein}
E.~Meckes.
\newblock On {S}tein’s method for multivariate normal approximation.
\newblock In \emph{High dimensional probability V: the Luminy volume}, pages
  153--178. Institute of Mathematical Statistics, 2009.

\bibitem[Morris et~al.(2008)Morris, Goodreau, Butts, Handcock, and
  Hunter]{morris2008ergm}
M.~Morris, S.~Goodreau, C.~Butts, M.~Handcock, and D.~Hunter.
\newblock ergm: {A} package to fit, simulate and diagnose exponential-family
  models for networks.
\newblock 2008.

\bibitem[Muandet et~al.(2017)Muandet, Fukumizu, Sriperumbudur, and
  Sch{\"o}lkopf]{muandet2017kernel}
K.~Muandet, K.~Fukumizu, B.~Sriperumbudur, and B.~Sch{\"o}lkopf.
\newblock Kernel mean embedding of distributions: A review and beyond.
\newblock \emph{Foundations and Trends{\textregistered} in Machine Learning},
  10\penalty0 (1-2):\penalty0 1--141, 2017.

\bibitem[Ospina-Forero et~al.(2019)Ospina-Forero, Deane, and
  Reinert]{ospina2019assessment}
L.~Ospina-Forero, C.~M. Deane, and G.~Reinert.
\newblock Assessment of model fit via network comparison methods based on
  subgraph counts.
\newblock \emph{Journal of Complex Networks}, 7\penalty0 (2):\penalty0
  226--253, 2019.

\bibitem[Ouadah et~al.(2020)Ouadah, Robin, and Latouche]{ouadah2020degree}
S.~Ouadah, S.~Robin, and P.~Latouche.
\newblock Degree-based goodness-of-fit tests for heterogeneous random graph
  models: Independent and exchangeable cases.
\newblock \emph{Scandinavian Journal of Statistics}, 47\penalty0 (1):\penalty0
  156--181, 2020.

\bibitem[Reinert and R{\"o}llin(2009)]{reinert2009multivariate}
G.~Reinert and A.~R{\"o}llin.
\newblock Multivariate normal approximation with stein’s method of
  exchangeable pairs under a general linearity condition.
\newblock \emph{The Annals of Probability}, 37\penalty0 (6):\penalty0
  2150--2173, 2009.

\bibitem[Reinert and R{\"o}llin(2010)]{reinert2010random}
G.~Reinert and A.~R{\"o}llin.
\newblock Random subgraph counts and $u$-statistics: multivariate normal
  approximation via exchangeable pairs and embedding.
\newblock \emph{Journal of Applied Probability}, 47\penalty0 (2):\penalty0
  378--393, 2010.

\bibitem[Reinert and Ross(2019)]{reinert2019approximating}
G.~Reinert and N.~Ross.
\newblock Approximating stationary distributions of fast mixing {G}lauber
  dynamics, with applications to exponential random graphs.
\newblock \emph{The Annals of Applied Probability}, 29\penalty0 (5):\penalty0
  3201--3229, 2019.

\bibitem[Rolls et~al.(2015)Rolls, Wang, McBryde, Pattison, and
  Robins]{rolls2015simulation}
D.~A. Rolls, P.~Wang, E.~McBryde, P.~Pattison, and G.~Robins.
\newblock A simulation study comparing epidemic dynamics on exponential random
  graph and edge-triangle configuration type contact network models.
\newblock \emph{PloS one}, 10\penalty0 (11):\penalty0 e0142181, 2015.

\bibitem[Schmid and Desmarais(2017)]{schmid2017exponential}
C.~S. Schmid and B.~A. Desmarais.
\newblock Exponential random graph models with big networks: Maximum
  pseudolikelihood estimation and the parametric bootstrap.
\newblock In \emph{2017 IEEE International Conference on Big Data}, pages
  116--121. IEEE, 2017.

\bibitem[Schweinberger(2012)]{schweinberger2012statistical}
M.~Schweinberger.
\newblock Statistical modelling of network panel data: {G}oodness of fit.
\newblock \emph{British Journal of Mathematical and Statistical Psychology},
  65\penalty0 (2):\penalty0 263--281, 2012.

\bibitem[Shervashidze et~al.(2011)Shervashidze, Schweitzer, Leeuwen, Mehlhorn,
  and Borgwardt]{shervashidze2011weisfeiler}
N.~Shervashidze, P.~Schweitzer, E.~J.~v. Leeuwen, K.~Mehlhorn, and K.~M.
  Borgwardt.
\newblock Weisfeiler-lehman graph kernels.
\newblock \emph{Journal of Machine Learning Research}, 12\penalty0
  (Sep):\penalty0 2539--2561, 2011.

\bibitem[Shore and Lubin(2015)]{shore2015spectral}
J.~Shore and B.~Lubin.
\newblock Spectral goodness of fit for network models.
\newblock \emph{Social Networks}, 43:\penalty0 16--27, 2015.

\bibitem[Snijders(2002)]{snijders2002markov}
T.~A. Snijders.
\newblock Markov chain {M}onte {C}arlo estimation of exponential random graph
  models.
\newblock \emph{Journal of Social Structure}, 3\penalty0 (2):\penalty0 1--40,
  2002.

\bibitem[Snijders et~al.(2006)Snijders, Pattison, Robins, and
  Handcock]{snijders2006new}
T.~A. Snijders, P.~E. Pattison, G.~L. Robins, and M.~S. Handcock.
\newblock New specifications for exponential random graph models.
\newblock \emph{Sociological Methodology}, 36\penalty0 (1):\penalty0 99--153,
  2006.

\bibitem[Sriperumbudur et~al.(2011)Sriperumbudur, Fukumizu, and
  Lanckriet]{sriperumbudur2011universality}
B.~K. Sriperumbudur, K.~Fukumizu, and G.~R. Lanckriet.
\newblock Universality, characteristic kernels and {RKHS} embedding of
  measures.
\newblock \emph{Journal of Machine Learning Research}, 12\penalty0
  (Jul):\penalty0 2389--2410, 2011.

\bibitem[Steglich et~al.(2006)Steglich, Snijders, and
  West]{steglich2006applying}
C.~Steglich, T.~A. Snijders, and P.~West.
\newblock Applying siena.
\newblock \emph{Methodology}, 2\penalty0 (1):\penalty0 48--56, 2006.

\bibitem[Stein(1986)]{stein1986approximate}
C.~Stein.
\newblock Approximate {C}mputation of {E}xpectations.
\newblock IMS, 1986.

\bibitem[Sugiyama and Borgwardt(2015)]{sugiyama2015halting}
M.~Sugiyama and K.~Borgwardt.
\newblock Halting in random walk kernels.
\newblock In \emph{Advances in Neural Information Processing Systems}, pages
  1639--1647, 2015.

\bibitem[Sugiyama et~al.(2018)Sugiyama, Ghisu, Llinares-L{\'o}pez, and
  Borgwardt]{sugiyama2018graphkernels}
M.~Sugiyama, M.~E. Ghisu, F.~Llinares-L{\'o}pez, and K.~Borgwardt.
\newblock graphkernels: {R} and {P}ython packages for graph comparison.
\newblock \emph{Bioinformatics}, 34\penalty0 (3):\penalty0 530--532, 2018.

\bibitem[Vishwanathan et~al.(2010)Vishwanathan, Schraudolph, Kondor, and
  Borgwardt]{vishwanathan2010graph}
S.~V.~N. Vishwanathan, N.~N. Schraudolph, R.~Kondor, and K.~M. Borgwardt.
\newblock Graph kernels.
\newblock \emph{The Journal of Machine Learning Research}, 11:\penalty0
  1201--1242, 2010.

\bibitem[Wasserman and Faust(1994)]{wasserman1994social}
S.~Wasserman and K.~Faust.
\newblock \emph{Social {N}etwork {A}nalysis: {M}ethods and {A}pplications},
  volume~8.
\newblock Cambridge University Press, 1994.

\bibitem[Xu and Matsuda(2020)]{xu2020stein}
W.~Xu and T.~Matsuda.
\newblock A {S}tein goodness-of-fit test for directional distributions.
\newblock \emph{The 23rd International Conference on Artificial Intelligence
  and Statistics}, 2020.

\bibitem[Yang et~al.(2018)Yang, Liu, Rao, and Neville]{yang2018goodness}
J.~Yang, Q.~Liu, V.~Rao, and J.~Neville.
\newblock Goodness-of-fit testing for discrete distributions via {S}tein
  discrepancy.
\newblock In \emph{International Conference on Machine Learning}, pages
  5557--5566, 2018.

\bibitem[Yang et~al.(2019)Yang, Rao, and Neville]{yang2019stein}
J.~Yang, V.~Rao, and J.~Neville.
\newblock A {S}tein--{P}apangelou goodness-of-fit test for point processes.
\newblock In \emph{The 22nd International Conference on Artificial Intelligence
  and Statistics}, pages 226--235, 2019.

\bibitem[Yin et~al.(2019)Yin, Phillips, and Butts]{yin2019selection}
F.~Yin, N.~E. Phillips, and C.~T. Butts.
\newblock Selection of exponential-family random graph models via {H}eld-{O}ut
  {P}redictive {E}valuation ({HOPE}).
\newblock \emph{arXiv preprint arXiv:1908.05873}, 2019.

\bibitem[Zhang et~al.(2008)Zhang, Friend, Traud, Porter, Fowler, and
  Mucha]{zhang2008community}
Y.~Zhang, A.~J. Friend, A.~L. Traud, M.~A. Porter, J.~H. Fowler, and P.~J.
  Mucha.
\newblock Community structure in congressional cosponsorship networks.
\newblock \emph{Physica A: Statistical Mechanics and its Applications},
  387\penalty0 (7):\penalty0 1705--1712, 2008.

\end{thebibliography}

\newpage

\appendix
\onecolumn

\begin{center}
    \Large A Stein Goodness-of-test for Exponential Random Graph Models\\
    Supplementary Material
\end{center}

\section{Proofs and Additional Lemmas}\label{supp:proofs}

{\bf{Proof of Theorem \ref{thm:asy_null_first}}}

\medskip
{For {convenience} we re-state the theorem here.}

\medskip 
{\bf{Theorem \ref{thm:asy_null_first}.}}
{\it 
Let $q(x)=\operatorname{ERGM}(\beta, t)$ {satisfy} Assumption \ref{assum:er_approx} and let ${\tilde q}$ denote the distribution of {ER$(a^*)$.}
For $f\in \H$ equipped with kernel $K$, let
$
f_x^*(\cdot) = \frac{ (\A_q - \A_{\tilde q} ) K(x,\cdot)}{\left\|(\A_q - \A_{\tilde q} ) K(x,\cdot) \right\|_{\H}}.
$
Then there is an explicit constant $C=C(\beta, t, {K})$  such that for all $\epsilon > 0,$
\begin{eqnarray*}
{
\P ( |  \gKSS (q,X)  - \gKSS ( {\tilde q}, Y) | \, >  \,  \epsilon)}\le \Big\{  || \Delta (\gKSS(q, \cdot))^2 || ( 1 +  || \Delta \gKSS(q, \cdot) || ) +  4 \sup_x 
(|| \Delta f_x^*||^2)  \Big\} 
{n \choose 2}  \frac{C}{{\epsilon^2 \sqrt{n}}} . 
\end{eqnarray*}
}

Under the null hypothesis, $X\sim q$ which is an  ERGM satisfying Assumption \ref{assum:er_approx}. 
{Let $Y\sim \tilde{q}$, where $\tilde{q}$ is the Bernoulli random graph  with edge probability $a^*$ and $a^*$ is {a}  solution to the equation in Assumption \ref{assum:er_approx}.
{We use the triangle inequality, 
\begin{equation}\label{th1proof}
 | \gKSS (q,x)  -  \gKSS (\tilde{q},y)| \le  
  | \gKSS (q,x)  -  \gKSS (\tilde{q},x)|  + |   \gKSS (\tilde{q},y)|,x)  -  \gKSS (\tilde{q},y)|.
\end{equation} 
This gives rise to two approximation terms.}
For the first summand in \eqref{th1proof}, we start with noting that
$$
 \gKSS (q,x) = \sup_{f \in \H, || f|| \le 1} | \A_q f (x) |
 =  \sup_{f \in \H, || f|| \le 1} | ( \A_q -  \A_{\tilde q}  +  \A_{\tilde q}) f (x) |
 \le \sup_{f \in \H, || f|| \le 1} | ( \A_q -  \A_{\tilde q} ) f(x) | + \gKSS ({\tilde q}, x) 
$$
{and this inequality also holds with the roles of $q$ and $\tilde{q}$ reversed,} 
so that 
$$ |  \gKSS (q,x) - \gKSS ( {\tilde q}, x) | \le \sup_{f \in \H, || f|| \le 1} \left| ( \A_q -  \A_{\tilde q} ) f(x) \right| {\quad  =   \sup_{f \in \H, || f|| \le 1} \left| \langle f (\cdot), ( \A_q -  \A_{\tilde q} ) k(x, \cdot) \rangle_{\H}   \right| }  $$
{where we used that due to the RKHS property, $f(x) = \langle f (\cdot), k(x, \cdot)\rangle_{\H}$. Thus} we have an explicit form for the optimal $f_x^* $ in this expression, namely
$  f_x^* (\cdot) = {   ( \A_q -  \A_{\tilde q} )k(x, \cdot)} / {\| ( \A_q -  \A_{\tilde q} )k(x, \cdot) \|_{\H} },$
 }
and   
  $$ |  \gKSS (q,x) - \gKSS ( {\tilde q}, x) | \le   | ( \A_q -  \A_{\tilde q} ) f_x^*(x)  |.$$ 
Following the steps for the proof of Theorem 1.7 in \cite{reinert2019approximating} but working directly with a function $f$ without using that it is a solution of a Stein equation, it is straightforward to show that for all $f \in \H,$ it holds that for $Y \sim { \tilde q},  $
$$ | \E (\A_q f(Y) - \A_{\tilde{q}} f(Y) ) | \le || \Delta f|| {n \choose 2} \frac{C(\beta, t)}{\sqrt{n}}$$ 
for an explicit constant $C$ which depends only on the vectors $\beta$ and $t$. 
Moreover inspecting the proof of Lemma 2.4 in \cite{reinert2019approximating} the bound is indeed a stronger bound,
$$  \frac{1}{N} \sum_{s \in {N}} \E | (\A_q^{(s)}  f(Y) - \A_{\tilde{q} } ^{(s)} f(Y) ) | \le || \Delta f|| {n \choose 2} \frac{C(\beta, t)}{\sqrt{n}}.$$
In particular with the crude bound $| ( \A_q^{(s)}   - \A_{\tilde{q}}^{(s)} ) f| \le 2  || \Delta f  || $ it follows that
$$
\E  \left\{ \left( { \frac{1}{N} \sum_{s \in {N}} } (\A_q^{(s)}  f(Y) - \A_{\tilde{q}}^{(s)} f(Y) ) \right)^2 \right\}  \le 2  || \Delta f  ||^2 
 {n \choose 2} \frac{C(\beta, t)}{\sqrt{n}}.
$$
Thus, using the Chebychev inequality,  for all $\epsilon > 0$, 
$$ \P ( |   ( \A_q -  \A_{\tilde q} ) f_Y^* (Y)  | > \epsilon)
\le \frac{1}{\epsilon^2} \Var ( ( \A_q -  \A_{\tilde q} ) f_Y^* (Y) )
\le  
4 \sup_x (|| \Delta f_x^*||^2)  
{n \choose 2} \frac{ { {C}}(\beta, t)}{\epsilon^2 \sqrt{n}}.
$$ 
Hence 
$$ \P ( |  \gKSS (q,Y) - \gKSS ( {\tilde q}, Y) | > \epsilon)
\le  {4 \sup_x (|| \Delta f_x^*||^2)  }  {n \choose 2} \frac{ {{C}}(\beta, t)}{ \epsilon^2 \sqrt{n}}.
$$ 
{For the second summand in Eq.\eqref{th1proof}, to} bound $ |\gKSS (q,X) -  \gKSS (q,Y) | $ we consider the test function $h(x) =  \gKSS(q, x)$ and apply Theorem 1.7 from \cite{reinert2019approximating} to give that 
$$ | \E (  \gKSS (q,X) -   \gKSS (q,Y) ) | \le || \Delta \gKSS(q, \cdot) ||   {n \choose 2}  \frac {{{\tilde{C}}}}{\sqrt {n}}. $$
{Here ${\tilde C}$ is a new constant which depends only on $\beta$ and $t$.}
Similarly we can approximate the square of the expectation using that $(a-b)^2 = a^2 - b^2 + 2 b(b-a) $ {and write 
$$
\E \{ (  \gKSS (q,X) -   \gKSS (q,Y) )^2 \}
= \E \{ (  \gKSS (q,X)^2 \}  -  \E \{ \gKSS (q,Y)^2 \} + 2 \E \{  \gKSS (q,X) ( \gKSS (q,X) - \gKSS (q,Y) \}.
$$ The first summand can be bounded with  Theorem 1.7 from \cite{reinert2019approximating} using the test function 
$h(x) = \gKSS(q, x)^2 $. For the second summand, we the Cauchy-Schwarz inequality gives 
$$
| \E  \{  \gKSS (q,X) ( \gKSS (q,X) - \gKSS (q,Y) \} | \le [ \E (\gKSS (q,X)^2 ) ]^\frac12 
[ \E \{ (  \gKSS (q,X)   -   \gKSS (q,Y) )^2 \} ]^\frac12 
$$ 
As $| \gKSS(q, x)| \le 1$ we obtain
$$
\E \{ (  \gKSS (q,X) -   \gKSS (q,Y) )^2 \} \le \E \{  \gKSS (q,X)^2 \}  -  \E \{ \gKSS (q,Y)^2 \}
+ 2 [ \E \{ (  \gKSS (q,X)   -   \gKSS (q,Y) )^2 \} ]^\frac12 .
$$
Solving this quadratic inequality gives 
$$
\E \{ (  \gKSS (q,X) -   \gKSS (q,Y) )^2 \} \le \left( 1 - \sqrt{1 - (\E \{   \gKSS (q,X)^2 \}  -  \E \{ \gKSS (q,Y)^2 \} )}  \right) ^2 
$$
and $ | \E \{   \gKSS (q,X)^2 \}  -  \E \{ \gKSS (q,Y)^2  | \le 1 $ we obtain that 
$$
\E \{ (  \gKSS (q,X) -   \gKSS (q,Y) )^2 \} \le |\E \{   \gKSS (q,X)^2 \}  -  \E \{ \gKSS (q,Y)^2 \} | .
$$
With Theorem 1.7 from \cite{reinert2019approximating} for the  test function 
$h(x) = \gKSS(q, x)^2 $ we obtain 
} 
$$ \E \{ (  \gKSS (q,X) -   \gKSS (q,Y) ) \}^2 \le ( || \Delta (\gKSS(q, \cdot))^2 ||  {n \choose 2}  \frac{{\hat C}}{ {\sqrt{n}} } ,$$
where ${\hat C} $ is another constant which depends only on $\beta$ and $t$ {but not on $n$}. With the Chebychev inequality and the triangle inequality we conclude that there is an {explicitly computable} constant $C$ such that for all $x$
$$
\P ( | \gKSS (q,X)  - \gKSS ( {\tilde q}, Y) | > \epsilon)
\le ( \sup_x (|| \Delta f_x^*||^2)  +  || \Delta (\gKSS(q, \cdot))^2 || ( 1 +  || \Delta \gKSS(q, \cdot) || ) 
{n \choose 2}  \frac{C}{ {\epsilon^2}{\sqrt{n}}} . 
$$ 
The assertion follows.  $\hfill \Box$
%


\medskip 
For the approximate distribution of $\gKSS( {\tilde q}, Y)$ it is more convenient to consider the square as given in Eq.\eqref{eq:gkss_quadratic_form}; this is addressed by  Theorem \ref{normalapprox}.

\medskip
{\bf Proof of Theorem \ref{normalapprox}}

{For {convenience} we re-state the assumptions and the theorem here.}

To approximate the distribution of $\gKSS^2$ under the null hypothesis we make the following assumptions (Assumption \ref{assum:er_approx_kernel} in the main text) on the kernel $K$ for the RKHS $\H$, namely that for $x, y \in \{ 0, 1 \}^N,$ \vspace{-2mm} 
\begin{enumerate}[i)]
    \item $\H$ is a tensor product RKHS, $\H = \otimes_{s \in [n]}\H_s $; \vspace{-2mm} 
    \item $k$ is a product kernel, $k(x, y) = \otimes_{s \in [N]} l_s(x_s, y_s)$; \vspace{-2mm} 
    \item  $ \langle l_s (x_s, \cdot), l_s (x_s, \cdot) \rangle_{\H_s}  =1$; \vspace{-2mm} 
    \item  $l_s(1, \cdot) - l_s(0, \cdot) \ne 0$ for all $s \in [N]$. \vspace{-2mm} 
    \end{enumerate}

{\bf{Theorem \ref{normalapprox}.}} 
{\it 
 Assume that the conditions i) - iv) in Assumption \ref{assum:er_approx_kernel} hold. 
 Let $\mu = \E [  \operatorname{gKSS}^2 ( {\tilde q}, Y)] $ and $\sigma^2 = \Var [ \operatorname{gKSS}^2 ( {\tilde q}, Y)]. $ Set 
 $W = \frac{1}{\sigma} (  \operatorname{gKSS}^2 ( {\tilde q}, Y)]  - \mu)$ and let  $Z$ denote a standard normal variable, Then there is an explicit constant $C = C(a^*, l_s, s \in [N]) $ such that 
\begin{equation*} 
    || {\mathcal L}(W) - {\mathcal L}(Z)||_1  \le \frac{C}{\sqrt{N}}. 
    \end{equation*} }

\medskip
For the Bernoulli random graph distribution ${\tilde q},$ and $s \in [N]$,
$$
\A_{\tilde q}^{(s)} f(x)  =
a^* f(x^{(s,1)} - f(x) ) + (1-a^*)  f(x^{(s,0)} - f(x) ).
$$
Thus, 
\begin{eqnarray*}
\gKSS^2({\tilde q}, x) &=& 
\frac{1}{N^2} \sum_{s,s' \in [N]} \left\langle  \A_{\tilde q}^{(s)}  K(x, \cdot) ,   \A_{\tilde q}^{(s')} K(x, \cdot) \right\rangle \\
&=& \frac{1}{N^2} \sum_{s,s' \in [N]} \Big\langle a^* \left(  K(x^{(s,1)}, \cdot) -  K(x, \cdot) \right) + (1-a^*) \left(K(x^{(s,0)}, \cdot) - K(x, \cdot) \right), \\
&& \quad \quad \quad \quad \quad  \quad \quad  a^* \left(  K(x^{(s',1)}, \cdot) - K(x, \cdot) \right) + (1-a^*) \left( K(x^{(s',0)}, \cdot) - K(x, \cdot) \right)\Big\rangle
.
\end{eqnarray*}
Under Assumptions \eqref{ass21}, \eqref{ass22} and \eqref{ass23} we can write 
\begin{eqnarray*}
 K(x^{(s,1)}, \cdot) -  K(x, \cdot) &=& 
 (l_s(1 ,  \cdot) - l_s( x_s,   \cdot) ) \prod_{t \ne s} l_t (x_t,  ,  \cdot)\\
 &=& (1 - x_s)  (l_s(1, \cdot ) - l_s( 0 ,  \cdot) ) l_{s'} (x_{s'}, \cdot ) \prod_{t \ne s, s'} l_t (x_t,    \cdot).
\end{eqnarray*}
Similarly, 
\begin{eqnarray*}
 K(x^{(s,0)}, \cdot) -  K(x, \cdot) &=& 
 - x_s  (l_s(1, \cdot ) - l_s( 0 ,  \cdot) ) l_{s'} (x_{s'}, \cdot) \prod_{t \ne s, s'} l_t (x_t,    \cdot).
\end{eqnarray*}
Abbreviating $g(x^{-s,s'}, \cdot) := \prod_{t \ne s, s'} l_t (x_t,    \cdot)$ we obtain that
\begin{eqnarray*}
\gKSS^2({\tilde q}, x) 
&=& \frac{1}{N^2} \sum_{s,s' \in [N]}  (a^* (1 - x_s) - (1-a^*) x_s)  (a^* (1 - x_s') - (1-a^*) x_s')\\
&& \langle (l_s(1, \cdot ) - l_s( 0 ,  \cdot)  l_{s'} (x_{s'}, \cdot), 
\langle (l_s'(1, \cdot ) - l_s'( 0 ,  \cdot) )  l_{s} (x_{s}, \cdot) \rangle 
\langle g(x^{-s,s'} , \cdot) ,  g(x^{-s,s'} , \cdot) \rangle  \\
&=& \frac{1}{N^2} \sum_{s,s' \in [N]}  (a^*  - x_s)   (a^*  - x_s') 
\langle (l_s(1, \cdot ) - l_s( 0 ,  \cdot) ) l_{s'} (x_{s'}, \cdot), 
\langle (l_s'(1, \cdot ) - l_s'( 0 ,  \cdot) )  l_{s} (x_{s}, \cdot) \rangle \\
&=& \frac{1}{N^2} \sum_{s,s' \in [N]}  (a^*  - x_s)   (a^*  - x_s')  \langle l_{s} (x_{s}, \cdot) ,   l_{s'} (x_{s'}, \cdot) \rangle c(s,s') 
\end{eqnarray*}
with
$$c (s,s') = \langle l_s (1, \cdot ) - l_s( 0 ,  \cdot)  , 
 l_{s'}(1, \cdot ) - l_{s'}( 0 ,  \cdot) \rangle
$$
{not depending on $x$. Here we used that by assumptions \eqref{ass22}  and \eqref{ass23}, $\langle g(x^{-s,s'} , \cdot) ,  g(x^{-s,s'} , \cdot) \rangle = 1.$}
Thus, when replacing $x$ by $Y$, a random vector in $\{0,1\}^N$ representing a Bernoulli random graph on $n$ vertices with edge probability $p$, then $\gKSS^2({\tilde q}, Y)$ is an average of locally dependent random variables. Hence, {using Stein's method} we obtain a  normal approximation {with bound}, as follows. Let 
$ {\mathcal I} = \{ (s,s'): s, s' \in [N] \}$
so that $|  {\mathcal I}| = N^2.$ For $\alpha = (s,s') \in  {\mathcal I}$ set
$$X_\alpha =  \frac{1}{N^2} (a^*  - Y_s)   (a^*  - Y_{s'})  \langle l_{s} (Y_{s}, \cdot) ,   l_{s'} (Y_{s'}, \cdot) \rangle c(s,s') ;$$
then 
$$\gKSS^2({\tilde q}, Y) =  \sum_{ \alpha \in  {\mathcal I}}   X_\alpha$$
and unless $\alpha$ and $ \beta$ share at least one vertex, the random variables $X_\alpha$ and $X_\beta$ are independent. 
Let $\mu_\alpha = \E X_\alpha$ and $\sigma^2 = \Var ( \gKSS^2({\tilde q}, Y) )) $; these quantities depend on the chosen kernels $l_s$. {We use the standardised count}  
$$W =  \sum_{\alpha \in  {\mathcal I}} \frac{X_\alpha - \mu_{\alpha}}{\sigma}
= \frac{1}{\sigma}   \gKSS^2({\tilde q}, Y) - \sum_{\alpha \in  {\mathcal I}} \frac{\mu_{\alpha}}{\sigma};$$
then $W$ has mean zero, variance 1, and results from Section 4.7 in \cite{chen2010} apply. In their notation, with 
$A_{(s,s')} = \{ \beta = (t,t') \in  {\mathcal I} : | \{ s,s' \} \cap \{t,t' \}| \ne \emptyset$, condition (LD1) is satisfied. Applying Theorem 4.13, p.134,  from \cite{chen2010} yields that, with $|| \cdot ||_1$ denoting $L_1$-distance, $\mathcal L$ denoting the law of a random variable, and $Z$ denoting a standard normal variable,
\begin{equation} \label{L1approx}
    || {\mathcal L}(W) - {\mathcal L}(Z)||_1 \le \sqrt{\frac{2}{\pi}} \E \left| \sum_{\alpha \in  {\mathcal I}}
    ( \xi_\alpha \eta_\alpha - \E ( \xi_\alpha \eta_\alpha)) 
    \right| + \sum_{\alpha \in  {\mathcal I}} \E | \xi_\alpha \eta_\alpha^2| 
    \le \sqrt{\frac{2}{\pi}}  \sqrt{ \Var ( \sum_{\alpha \in  {\mathcal I} } \xi_\alpha \eta_\alpha ) } + \sum_{\alpha \in  {\mathcal I}} \E | \xi_\alpha \eta_\alpha^2|.
\end{equation}
with $\xi_\alpha = ( X_\alpha - \mu_{\alpha} ) / {\sigma} $
and $\eta_\alpha = \sum_{\beta \in A_\alpha} X_\beta.$

{To obtain the dependence of the bound on $N$ we} 
assess its magnitude.
{First note that $| A_\alpha| \le 2 N$.} Using {that by the assumption \eqref{ass23},}
$ || l_s ||^2 =1$ for $s \in [N]$ and that $| a^* - Y_s| \le 1$ we can use the crude bounds $| c(s,s') | \le 4, $ so that $| X_\alpha| \le \frac{4}{N^2}  $ and $\mu_\alpha \le \frac{4}{N^2}  $, In particular, 
$| \xi_\alpha | \le \frac{8}{N^2 \sigma} $ and $ | \eta_\alpha| \le \frac{16}{N \sigma}$.
Thus, 
$$
\sum_{\alpha \in  {\mathcal I}} \E | \xi_\alpha \eta_\alpha^2| 
\le N^2 \times \frac{8}{N^2 \sigma } \times \frac{256}{N^2 \sigma^2} = \frac{2048}{N^2 \sigma^3}
.
$$

\medskip 
To evaluate the variance $\sigma^2$,
{$$
\sigma^2 = \sum_{\alpha \in {\mathcal I}} \Var X_\alpha + 
\sum_{\alpha \in {\mathcal I}} \sum_{\beta \in A_\alpha } Cov ( X_\alpha, X_\beta). 
$$
We evaluate these terms in turn. First, if $\alpha = (s,s)$ then 
$$  \Var X_\alpha
\le \frac{c(s,s)^2}{N^4} a^* ( 1 - a^*) 
$$
and if $\alpha = (s,s') $ with $s \ne s'$ then as 
$\langle l_s(x, \cdot), l_s(y, \cdot) \rangle \le 1$ from the assumption \eqref{ass23} and the Cauchy-Schwarz inequality, 
$$  \Var X_\alpha
\le  \E [ X_\alpha^2] \le \frac{c(s,s)^2}{N^4} [ a^* ( 1 - a^*) ]^2.
$$
Thus, 
$$ \sum_{\alpha \in {\mathcal I}} \Var X_\alpha \le  \frac{c(s,s)^2}{N^2} a^* ( 1 - a^*) . $$
Moreover, if $\alpha = (s,s)$ and $\beta = (s,t) \in A_\alpha$ then 
$$ | Cov (X_\alpha, X_\beta) | 
= \left| \frac{c(s,s) c(s,t) }{N^4}  \E \{ (a^*-Y_s)^3 (a^* - Y_t)  \langle l_{s} (Y_{s}, \cdot) ,   l_{t} (Y_{t}, \cdot) \rangle \} - \mu_\alpha \mu_\beta \right|
\le  2 \frac{|c(s,s) c(s,t) |}{N^4} 
$$ 
and there are order $N^2$ such terms $(\alpha, \beta)$ in the variance.}
The main contributions to the variance stem from $Cov(X_\alpha, X_\beta)$ for $\beta \in {\mathcal I}_\alpha$ and for $\alpha = (s.s') $ with $s \ne s'$.  Assumption \eqref{ass24}  guarantees that $c(s,s') \ne 0$.  Then for $\beta = (s, t),$ with $t \ne s$,
\begin{eqnarray*}  
Cov(X_\alpha, X_\beta) &=&  \frac{1}{N^4} c(\alpha) c(\beta) 
\E (a^* - Y_s)^2 (a^* - Y_s') (a^* - Y_t) 
 \langle l_{s} (Y_{s}, \cdot) ,   l_{s'} (Y_{s'}, \cdot) \rangle  \langle l_{s} (Y_{s}, \cdot) ,   l_{t} (Y_{t}, \cdot) \rangle  \\
 && 
 - \frac{1}{N^4}  (a^*)^4( 1-a^*)^4 c(s,s') c(s,t)  
\end{eqnarray*} 
and expanding the expectation gives a contribution of the order $N^{-4}$. The overall contribution of such covariance terms, of which there are order $N^3$, to the variance is hence of order $N^{-1}$, and therefore $\sigma^2$ is of order $N^{-1}$ and $\sigma $ is of order $\sqrt{N}$. 

Similarly, $$ \Var \left( \sum_{\alpha \in  {\mathcal I} } \xi_\alpha \eta_\alpha \right) = 
\Var \left( \sum_{\alpha \in  {\mathcal I} } \sum_{\beta \in A_\alpha }\xi_\alpha \xi_\beta  \right)  
=  \sum_{\alpha \in  {\mathcal I} } \sum_{\beta \in A_\alpha }
 \sum_{\gamma \in  {\mathcal I} } \sum_{\delta \in A_\gamma } Cov (\xi_\alpha \xi_\beta , \xi_\gamma \xi_\delta) 
 $$ is dominated by the covariances between $\xi_\alpha \xi_\beta$ and $\xi_\gamma \xi_\delta$ such that $\alpha$ and $\beta$ involve three distinct indices $s, s', t$, and $\gamma$ and $\delta$ involve three distinct indices $r, r', u$, and these two sets of three indices have non-zero intersection. These summands give a contribution of order $N^5 / (\sigma^4 N^8)$, which is of order $N^{-1}$, to the variance $ \Var \left( \sum_{\alpha \in  {\mathcal I} } \xi_\alpha \eta_\alpha \right) $.  {A crude bound is obtained as 
 $ \Var ( \sum_{\alpha \in  {\mathcal I} } \xi_\alpha \eta_\alpha )
 \le \frac{512}{\sigma^4 N^3}. 
 $
 These estimates give  that the bound in Eq.\eqref{L1approx} is of the order $N^{-\frac12}$.  {All moment expressions can be bounded explicitly and thus the constant $C$ can be computed explicitly. The conclusion follows.} 
}

\bigskip 
{\bf{Proof of Proposition \ref{bootstrapnormal}}}

\medskip 
{For convenience we re-state the result here again.}

{\bf Proposition \ref{bootstrapnormal}.}
{\it
Let    
$$Y  = \frac{1}{B^2}  \sum_{s, t \in [N] } ( k_s k_t - \E (k_s k_t) ) h_x (s, t).$$
Assume that  $h_x$ is  bounded such that $Var(Y)$ is non-zero. 
Then if $Z$ is mean zero normal with variance $Var(Y)$, there is an {explicitly computable}  constant $C>0$ such that for all three times  continuously differentiable functions $g$ with bounded derivatives up to order 3, 
$$
| \E [ g(Y) ] - \E [g(Z)] \le \frac{C}{B}. 
$$ 
}

\medskip 
 For normal approximation in the presence of weak dependence, Charles Stein \citep{stein1986approximate} introduced the method of exchangeable
pairs: construct a sum $W'$ such that $(W,W')$ form an exchangeable pair, and
such that $\IE^W(W'-W)$ is (at least approximately) linear in $W$. This
linearity condition arises naturally when thinking of correlated bivariate
normals. As a multivariate generalisation, \cite{reinert2009multivariate}
considered the general setting that 
\ben                                                                \label{1}
    \IE^W(W'-W) = - \Lambda W + R
\ee
for a matrix $\Lambda$ and a vector $R$ with small $\IE|R|$  is treated. In  a
followup paper~\citep{meckes2009stein} the results by 
\cite{chatterjee2008multivariate} and  \cite{reinert2009multivariate} are combined using
slightly different smoothness conditions on test functions as compared to
\cite{reinert2009multivariate}.
In \cite{reinert2009multivariate} it was found that a statistic of interest can often be embedded into a larger vector of statistics such that 
\eqref{1} holds with $R=0$; this embedding does not
directly correspond to Hoeffding projections, although it is related to the
latter. In \cite{reinert2010random} this embedding is applied to complete non-degenerate
U-statistics. among other examples. In this example the limiting covariance
matrix is not of full rank; yet the bounds on the normal approximation are of
the expected order.

The general setup is as follows. Denote by $W = (W_1,W_2,\dots,W_d)^t$ random vectors in $\IR^d$, where $W_i$ are
$\IR$-values random variables for $i=1,\dots,d$. We denote by $\Sigma$
symmetric, non-negative definite matrices, and hence by $\Sigma^{1/2}$ the
unique symmetric square root of~$\Sigma$. Denote by $\Id$ the identity matrix,
where we omit the dimension~$d$. Let  $Z$ denote a random
variable having standard $d$-dimensional multivariate normal distribution. We
abbreviate the transpose of the inverse of a matrix $\Lambda$ as $\Lambda^{-t}
:= (\Lambda^{-1})^t$.

For derivatives of smooth functions $h:\IR^d\to\IR$, we use the notation $\nabla
$ for the gradient operator. Denote by $\norm{\cdot}$ the supremum norm for both
functions and  matrices. If the corresponding derivatives exist for some
function $g:\IR^d\to\IR$, we abbreviate $\abs{g}_1
:=\sup_{i}\bnorm{\frac{\partial}{\partial x_i} g}$, $\abs{g}_2
:=\sup_{i,j}\bnorm{\frac{\partial^2}{\partial x_i \partial x_j}g}$, and so on.

The following result is shown in \cite{reinert2009multivariate}. 

\begin{theorem}[c.f. Theorem~2.1 \cite{reinert2009multivariate}] \label{thm1} Assume that
$(W,W')$ is an exchangeable pair of\/ $\IR^d$-valued random variables such that 
\ben								\label{2}
    \IE W = 0,\qquad\IE W W^t = \Sigma, 			
\ee 
with $\Sigma \in\IR^{d\times d}$ symmetric and positive definite. Suppose
further that~\eqref{1} is satisfied for an invertible matrix $\Lambda$ and a
$\sigma(W)$-measurable random variable~$R$. Then, if $Z$ has $d$-dimensional
standard normal distribution, we have for every three times differentiable function $g$
\beas                                                          \label{3}
    \babs{\IE g(W)-\IE g(\Sigma^{1/2}Z)}
    \leq 
    \frac{\abs{g}_2}{4}  I
    +\frac{\abs{g}_3}{12} II    
    +\bbklr{\abs{g}_1 +\ahalf d\norm{\Sigma}^{1/2} \abs{g}_2} III   ,  
\enas
where, with $\lambda^{(i)} = \sum_{m=1}^d\abs{(\Lambda^{-1})_{m,i}}$, 
\ba
I &= \sum_{i,j=1}^d{\lambda^{(i)}} \sqrt{\Var{\IE^W(W'_i-W_i)(W'_j-W_j)}},\\
II &= \sum_{i,j,k=1}^d{\lambda^{(i)}} \IE\babs{(W'_i-W_i)(W'_j-W_j)(W'_k-W_k)},\\
III &= \sum_{i}\lambda^{(i)}\sqrt{\IE R_i^2}.
\ee
\end{theorem}

Here we use the approach for statistics of the form 
 $   Y  = \frac{1}{B^2}  \sum_{s, t \in [N] } ( k_s k_t - \E (k_s k_t) ) h (s, t).$
The subscript $x$ is suppressed in $h_x$  to simplify notation. 
To apply Theorem \ref{thm1} we employ two additional statistics; including $Y$ as $W_1$, 
\beas 
   W_1  &=&\frac{1}{B^2}  \sum_{s, t \in [N] } ( k_s k_t - \E (k_s k_t) ) h (s, t) \nonumber\\
   W_2 &=& \frac{1}{B^2} \sum_{s, t \in [N] } ( k_s - \E (k_s ) ) h (s, t) \label{y3}\\
    W_3 &=& \frac{1}{B^2} \sum_{s \in [N] } ( k_s - \E (k_s ) ) h (s, s) \label{y4}.
\enas 
Given ${\bf{k}}=(k_1, \ldots, k_N)$  we construct an exchangeable pair $({\bf{k}}, {\bf{k'}})$ by
choosing an index $I \in [N]$ such that
$\IP (I=i) = \frac{k_i}{B}$
and if $I=i$ we set $k'_i = k_i - 1$ (we take a ball out of bin $i$ in the multinomial construction). 
Then we pick $J \in [N]$ uniformly and if $J=j$ we set $k_j' = k_j +1$ - we add the ball to bin $j$ which we took away from bin $i$.
All other $k_l's$ are left unchanged; $k_l' - k_l$ if $l \ne I, J$. Note that $I=J$ is possible in which case there is no change. 
Based on this exchangeable pair we set 
\beas 
   W_1'  &=& \frac{1}{B^2}\sum_{s, t \in [N] } ( k_s' k_t`- \E (k_s' k_t') ) h (s, t) \label{y1ex} \\
   W_2' &=& \frac{1}{B^2} \sum_{s, t \in [N] } ( k_s' - \E (k_s' ) ) h (s, t) \label{y3ex}\\
    W_3' &=& \frac{1}{B^2} \sum_{s \in [N] } ( k_s' - \E (k_s' ) ) h (s, s) \label{y4ex}.
\enas 
With $W= (W_1, W_2, W_3)$ and  $W'= (W_1', W_2',W_3')$ we have obtained an exchangeable pair $(W,W')$. Moreover $W$ has mean zero and finite covariance matrix. First we calculate $  \IE^W(W'-W)$ componentwise, starting with the easiest case to illustrate the argument. 
For this calculation we use that
\beas 
k'_I - k_I  &=& -1\\
k'_J - k_J &=& 1\\
 k_s' k_t`-  k_s k_t &=& (k_s' - k_s) (k_t' - k_t ) + k_s (k_t' - k_t) + k_t (k_s' - k_s). 
\enas 

Then, conditioning on $I$ and $J$, 
\beas 
\IE^W(W_3'-W_3) &=& \frac{1}{B^2} \sum_{s \in [N] } \IE^W( k_s' - k_s) h (s, s)\\
&=& \frac{1}{B^2}\frac{1}{BN} \sum_{s \in [N] }\sum_{i \in [N] } k_i \sum_{j \in [N] }  ( -  {\bf{1}}(s=i) h(i,i) + {\bf{1}}(s=j)  h (j,j) ) \\
&=&  - \frac{1}{B^2}\frac{1}{B} \sum_{i \in [N] } k_i h(i,i)  + \frac{1}{N}\frac{1}{B^2} \sum_{j \in [N] } h (j,j) \\
&=&  - \frac{1}{B^2} \frac{1}{B} \sum_{i \in [N] } (  k_i - \E (k_i)) h(i,i)\\
&=&  - \frac{1}{B} W_3. 
\enas 
Similar arguments yield 
$
\IE^W(W_2'-W_2) 
=  - \frac{1}{B} W_2. 
$
Finally, 
\beas 
\lefteqn{\IE^W(W_1'-W_1) }\\
 &=&\frac{1}{B^2}\sum_{s, t \in [N] }  \IE^W( k_s' k_t`-k_s k_t ) h (s, t)\\
&=& \frac{1}{B^2}\sum_{s, t \in [N] }  \IE^W [ (k_s' - k_s) (k_t' - k_t ) + k_s (k_t' - k_t) + k_t (k_s' - k_s)] h (s, t)\\
&=&\frac{1}{B^2} \sum_{s, t \in [N] }  \IE^W [ (k_s' - k_s) (k_t' - k_t ) ] h (s, t) + 2 \sum_{s, t \in [N] }  \IE^W [ k_t (k_s' - k_s) ] h (s, t).
\enas 
Here we used that $h(s,t) = h(t,s)$ in the last step. We tackle the conditional expectations separately. Again using  $h(s,t) = h(t,s)$, 
\beas
\lefteqn{ \sum_{s, t \in [N] } \IE^W [ (k_s' - k_s) (k_t' - k_t ) ] h (s, t) } \\
&=& \frac{1}{BN} \sum_{s,t \in [N] }\sum_{i \in [N] } k_i \sum_{j \in [N]} ({\bf{1}}(s=I, t=J)  + {\bf{1}}(s=J, t=I) 
[ (k_s' - k_s) (k_t' - k_t ) ] h (s, t)  \\
&=& -  \frac{2}{BN} \sum_{i \in [N] } k_i \sum_{j \in [N]}  {\bf{1}}( i \ne j)  h (i,j)  \\
&=& -  \frac{2}{BN} \sum_{i \in [N] } \sum_{j \in [N]}    k_i h (i,j)  +  \frac{2}{BN} \sum_{i \in [N] } k_i  h (i,i)  \\
&=& -  \frac{2}{BN} W_2 +  \frac{2}{BN} W_3. 
\enas 
Here the centering terms from $W_2$ and $W_3$ add up to 0 because the conditional expectation has mean zero, and are thus not included in the calculation. 

Moreover, 
\beas
\lefteqn{\sum_{s, t \in [N] }  \IE^W [  k_t (k_s' - k_s) ] h (s, t)}\\
&=& \frac{1}{BN} \sum_{s,t \in [N] }\sum_{i \in [N] } k_i \sum_{j \in [N]} \left( 
-  {\bf{1}}(s=i)  k_t h (i, t) + {\bf{1}}(s=j)   k_t h (j, t)  
\right) 
\\
&=& -  \frac{1}{B} \sum_{t \in [N] }\sum_{i \in [N] } k_i k_t  h (i, t) 
+ 
\frac{1}{N} \sum_{t \in [N] }\sum_{j \in [N]} 
   k_t h (j, t)  \\
   &=& -  \frac{1}{B} W_1 + \frac{1}{N}  W_2. 
\enas 
Hence
\beas 
\IE^W(W_1'-W_1)  &=&  -  \frac{2}{BN} W_2 +  \frac{2}{BN} W_3 -  \frac{2}{B} W_1 + \frac{2}{N}  W_2
\\
&=&  \frac{2}{BN} W_3+ \frac{2(B-1) }{BN}  W_2-  \frac{2}{B} W_1 .
\enas 
Hence \eqref{1} is satisfied with $R=0$ and 
\be
    \Lambda = \frac{1}{B}
    \begin{bmatrix}
  -2 & \frac{2(B-1)}{N} & \frac{2}{N} \\
 0 & -1 & 0 \\
0 & 0 & -1\\
    \end{bmatrix}
\ee
giving
$$\lambda^{(1)} = \frac{B}{2}; \lambda^{(2)} = B \frac{| N- B+1| }{N}; \lambda^{(3)} = \frac{B(N+1)}{N}. $$ 
{With} $B = FN $ we can bound 
$$\lambda^{(i)} \le  \max (F, 1/2)  B, \quad i=1, 2, 3.$$ 
To complete the argument we need to bound $I$ and $II$ from Theorem \ref{thm1}. 

\medskip To bound the conditional variance term $I$ from Theorem \ref{thm1}, 
\beas
I &= &\sum_{i,j=1}^3{\lambda^{(i)}} \sqrt{\Var{\IE^W(W'_i-W_i)(W'_j-W_j)}} \le {\max (F, 1/2) } \,   B \sum_{i,j=1}^3 \sqrt{\Var{\IE^W(W'_i-W_i)(W'_j-W_j)}}.
\enas
Instead of conditioning on $W$ we condition on ${\bf{k}}$ this conditioning would only increase the conditional variance. 
The largest variance contribution is from 
\beas
\lefteqn{\IE^{\bf{k}}(W'_1-W_1)^2 }\\
&=& \frac{1}{B^4} \sum_{s, t \in [N]} \sum_{u,v \in [N]} \IE^{\bf{k}} [ ( k_s' k_t'- k_s k_t) ( k_u'  k_v'- k_u k_v) h(s,t)  h(u,v)  ] \\
&=& \frac{1}{B^4}  \frac{1}{BN} \sum_{i \in [N]}  \sum_{j \in [N]} \sum_{s, t\in [N]}  \sum_{u,v\in [N]}  \IE^{\bf{k}} [ k_i \left( (k_s' - k_s) (k_t' - k_t ) + 2 k_s (k_t' - k_t) \right) \times \\
&&\qquad \qquad \qquad \qquad \qquad \qquad \qquad \quad  \left( (k_u' - k_u) (k_v' - k_v ) + 2 k_u (k_v' - k_v) \right)  h(s,t)  h(u,v)  ] \\
 &=& \frac{1}{B^4}  \frac{1}{BN} \sum_{i \in [N]}  \sum_{j \in [N]} \sum_{s, t\in [N]}  \sum_{u,v\in [N]}   \IE^{\bf{k}}[  k_i (k_s' - k_s) (k_t' - k_t )   (k_u' - k_u) (k_v' - k_v )   h(s,t)  h(u,v)  ] \\
 &&+ 2 \frac{1}{B^4}  \frac{1}{BN} \sum_{i \in [N]}  \sum_{j \in [N]} \sum_{s, t\in [N]}  \sum_{u,v\in [N]}  \IE^{\bf{k}} [  k_i  k_u (k_s' - k_s) (k_t' - k_t )     (k_v' - k_v) ] h(s,t)  h(u,v)  \\
 &&+ 2  \frac{1}{B^4}  \frac{1}{BN} \sum_{i \in [N]}  \sum_{j \in [N]} \sum_{s, t\in [N]}  \sum_{u,v\in [N]}  \IE^{\bf{k}} [  k_i k_s  (k_t' - k_t)  (k_u' - k_u) (k_v' - k_v )  h(s,t)  h(u,v)  ] \\
 &&+ 4  \frac{1}{B^4}  \frac{1}{BN} \sum_{i \in [N]}  \sum_{j \in [N]} \sum_{s, t\in [N]}  \sum_{u,v\in [N]}  \IE^{\bf{k}} [  k_i  k_s k_u  (k_t' - k_t)  (k_v' - k_v)  h(s,t)  h(u,v)  ] .
\enas 
Due to the exchangeable pair construction many sums simplify and the largest contribution to the variance is the last term; 
\beas
\lefteqn{4  \frac{1}{B^4}  \frac{1}{BN} \sum_{i \in [N]}  \sum_{j \in [N]} \sum_{s, t\in [N]}  \sum_{u,v\in [N]}  \IE^{\bf{k}} [  k_i  k_s k_u  (k_t' - k_t)  (k_v' - k_v)  h(s,t)  h(u,v)  ] }\\
&=& 4  \frac{1}{B^5 N}  \ \sum_{i \in [N]}  \sum_{j \in [N]} \sum_{s, t\in [N]}  \sum_{u,v\in [N]}  \IE^{\bf{k}} [  k_i  k_s k_u  (k_t' - k_t)  (k_v' - k_v)  h(s,t)  h(u,v)  ]  (  {\bf{1}}( t=i) +  {\bf{1}}( t=j)) \\ 
\\
&=& - 4  \frac{1}{B^5 N}  \ \sum_{i \in [N]}  \sum_{j \in [N]} \sum_{s \in [N]}  \sum_{u,v\in [N]}  \IE^{\bf{k}} [  k_i  k_s k_u  (k_v' - k_v)  h(s,i)  h(u,v)  ]  (  {\bf{1}}( v=i) +  {\bf{1}}( v=j)) \\  
&&+  4  \frac{1}{B^5 N}  \ \sum_{i \in [N]}  \sum_{j \in [N]} \sum_{s \in [N]}  \sum_{u,v\in [N]}  \IE^{\bf{k}} [  k_i  k_s k_u  (k_v' - k_v)  h(s,j)  h(u,v)  ]   (  {\bf{1}}( v=i) +  {\bf{1}}( v=j)) \\
&=& 4  \frac{1}{B^5 N}  \ \sum_{i \in [N]}  \sum_{j \in [N]} \sum_{s \in [N]}   k_i  k_s k_u   h(s,i)  h(u,i)   - 4  \frac{1}{B^5 N}  \ \sum_{i \in [N]}  \sum_{j \in [N]} \sum_{s \in [N]}   k_i  k_s k_u   h(s,i)  h(u,j)    \\  
&& -   4  \frac{1}{B^5 N}  \ \sum_{i \in [N]}  \sum_{j \in [N]} \sum_{s \in [N]}  \sum_{u \in [N]}   k_i  k_s k_u   h(s,j)  h(u,i)  +  4  \frac{1}{B^5 N}  B  \sum_{j \in [N]} \sum_{s \in [N]}  \sum_{u \in [N]}  k_s k_u    h(s,j)  h(u,j) .
\enas 
These terms have a variance contribution of order 
$ \frac{1}{B^{10} N^2}  \frac{B^6}{N^6} N^8 = \frac{1}{B^4}$ as long as $h(i,j)$ is bounded. 
The mixed variances in $I$ can be bounded using the Cauchy-Schwarz inequality. 
Overall the contribution to the term $I$ of Theorem \ref{thm1} is thus of order
$B \sqrt{ \frac{1}{B^4}} = \frac{1}{B}.$

\medskip For the term $II$ of Theorem \ref{thm1}, 
\beas
\sum_{a,b,c=1}^3{\lambda^{(a)}} \IE\babs{(W'_a-W_a)(W'_b-W_b)(W'_c-W_c)} 
&\le & {\max (F, 1/2) }\,  B \sum_{a,b,c=1}^3  \IE\babs{(W'_a-W_a)(W'_b-W_b)(W'_c-W_c)} .
\enas 
The largest contribution to this term is 

\beas 
\lefteqn{\IE\babs{(W'_1-W_1)^3}}\\
&\le & ||h||^3 \frac{1}{B^6} \sum_{s, t\in [N]}  \sum_{u,v\in [N]}  \sum_{x,y \in [N]} 
 \IE | ( k_s' k_t'-k_s k_t ) ( k_u' k_v'-k_u k_v )( k_x' k_y'-k_x k_y )|  \\
 &=&  ||h||^3  \frac{1}{B^6}  \frac{1}{BN} \sum_{i \in [N]}  \sum_{j \in [N]} \sum_{s, t\in [N]}  \sum_{u,v\in [N]}  \sum_{x,y \in [N]} \IE \left| k_i  \left( (k_s' - k_s) (k_t' - k_t ) + 2 k_s (k_t' - k_t) \right) \right. \\
 && \left.  \left( (k_u' - k_u) (k_v' - k_v ) + 2 k_u (k_v' - k_v) \right) \left( (k_x' - k_x) (k_y' - k_y ) + 2 k_x (k_y' - k_y) \right) \right| \\
 &\le &  ||h||^3  \frac{1}{B^6}  \frac{1}{BN} \sum_{i \in [N]}  \sum_{j \in [N]} \sum_{s, t\in [N]}  \sum_{u,v\in [N]}  \sum_{x,y \in [N]}  \IE [\left|  k_i    \left( (k_s' - k_s) (k_t' - k_t ) + 2 k_s (k_t' - k_t) \right) \right. \\
 && \left.  \left( (k_u' - k_u) (k_v' - k_v ) + 2 k_u (k_v' - k_v) \right) \left( (k_x' - k_x) (k_y' - k_y ) + 2 k_x (k_y' - k_y) \right) \right| .
\enas 
With $|| h|| = \max_{i,j} | h(i,j)|$ the leading term in this expression is 
\beas 
 \frac{8   ||h||^3 }{B^7 N}   \sum_{i \in [N]}  \sum_{j \in [N]} \sum_{s, t\in [N]}  \sum_{u,v\in [N]}  \sum_{x,y \in [N]}  \IE    \left|  k_i   k_s k_x (k_t' - k_t) k_u (k_v' - k_v)  (k_y' - k_y) \right| .
\enas 
Now,  not all of $t, v, y$ can be distinct for a non-zero contribution to this term; we can bound it by 
\beas
{ \frac{16   ||h||^3 }{B^7N }  
  \sum_{i, j, s, t, u, v, x  \in [N]}  \IE  k_i k_s k_u k_x  }
  ( {\bf{1}}( t=i) +  {\bf{1}}( t=j)) ( {\bf{1}}( v=i) +  {\bf{1}}( v=j))
 &\le& 
 \frac{64}{||h||^3 B^3}.
\enas 
Here we used that $\sum_i k_i = B$. All other cross-expectations can be bounded using the Cauchy-Schwarz inequality. Hence we conclude that the term $II$  in Theorem \ref{thm1} is of order $B^{-2}$. {All higher moments can be bounded explicitly and hence $C$ can be bounded explicitly.} The conclusion follows. 

\section{Graph Kernels}\label{supp:graph_kernel}

For a vertex-labeled graph {$x = \{ x_{ij} \}_{1 \le i , j \le n} \in \mathcal{G}^{lab}$}, with label range { $\{1, \ldots ,c \} = [c]$}, denote the vertex set by $V$, the edge set by $E$, and the label set by $\Sigma$. Consider {an} vertex-edge mapping $\psi :V \cup E \rightarrow [c]$. {In this paper we  use the following graph kernels.}

\paragraph{Vertex-Edge  Histogram Gaussian {Kernels}}
{The} \emph{vertex-edge label histogram} $h = (h^{111},h^{211},\dots, h^{ccc})$  ${= h(\psi, x) 
}$  has as components
$h^{l_1l_2l_3} = \left|\{v\in V, (v,u)\in E \, | \,  \psi(v,u) = l_1, \psi(u)= l_2,\psi(v)=l_3\}\right|$, for $l_1, l_2, l_3 \in [c]$;  {it is a combination of vertex label counts and edge label counts.} 
{Let 
$\langle h (x) , h (x') \rangle = \sum_{l_1,l_2,l_3} h (x)^{l_1,l_2,l_3}h {(x')}^{l_1,l_2,l_3}$.
{Following} \cite{sugiyama2015halting}, we define}
the vertex-edge histogram Gaussian (VEG) kernel between two graphs
{$x, x'$} 
 as 
$$
K_{\small VEG}(x, x';\sigma) = \exp{\left\{-\frac{\|h(x) - h(x')\|^2}{2\sigma^2}\right\}}. 
$$
{The VEG kernel is a special case of histogram-based kernels}  for assessing graph similarity {using feature maps, which are introduced} in \cite{kriege2016valid}.  {Adding a Gaussian RBF as in} \cite{sugiyama2015halting}, {yielding the VEG kernel, significantly improved problems such as classification accuracy, see} 
\citep{kriege2020survey}.
{In our implementation, as in \cite{sugiyama2018graphkernels}, $\psi$ is induced by the vertex index. If the vertices are indexed by $i \in [n]$  then  the label of vertex $v_i$ is 
$\psi(v_i)=i$; for edges,  $\psi(u,v)=1$ if $(u,v)\in E$ is an edge and $0$ otherwise.}

\paragraph{Geometric Random Walk Graph Kernels} 
A $k$-step random walk graph kernel \citep{sugiyama2015halting} is built as follows. Take $A_{\otimes}$
as the adjacency matrix of the direct (tensor) product $G_{\otimes} = (V_{\otimes},E_{\otimes},\psi_{\otimes})$ \citep{gartner2003graph} between $x$ and $x'$ such that {vertex labels match and edge labels match:} 
$$
V_{\otimes} = \{(v,v')\in V \times V' | \psi(v) = \psi'(v')\},
$$
$$
E_{\otimes} = \{((v,u),(v',u')))\in E \times E'\,  |  \, \psi (v,u) = \psi (v',u')\},
$$
and {use} the corresponding label mapping
$\psi_{\otimes}(v,v') = \psi(v) = \psi'(v')
$; $\psi_{\otimes}((v,v'),(u,u')) = \psi(v,u) = \psi'(v',u')
$. {With}  input parameters $(\lambda_0, \dots, \lambda_k)$, {the} $k-$step random walk kernel between two graphs
{$x, x'$}  is defined as 
$$
K_{\otimes}^{k}(x,x') = \sum_{i,j=1}^{ |V_{\otimes}|}\left[\sum_{t=0}^k \lambda_t A_{\otimes}^{\top}\right]_{i,j}.
$$

A geometric random walk kernel between two graphs
{$x, x'$}  takes the $\lambda$-weighted infinite sum from the $k$ step random walk kernels:
$$
K_{GRW}(x, x' ) =  \sum_{i,j=1}^{ |V_{\otimes}|} \left[(I - \lambda A_{\otimes})^{-1}\right]_{i, j} . 
$$
{In our implementation we choose, $\lambda_l = \lambda, \forall l=1,\dots, k$ and $\lambda=\frac{1}{3}$.}

\paragraph{Shortest Path Graph Kernels} {Shortest Path Graph Kernels, introduced by  \cite{borgwardt2005shortest},  are based on a transformation of the graph $x$, the Floyd transformation. The Floyd
transformation $F$  turns the original graph into the so-called shortest-path graph $y=F(x)$;  the  graph $y$ is a complete graph with vertex set $V$ with each edge labelled by the shortest distance in $x$ between the vertices on either end of the edge.} 
{For two networks  $x$ and $x'$}  the 1-step random walk kernel { $K^1_{\otimes}$} between the shortest-path graphs {$y = F(x)$ and $y'=F(x')$ } gives the  shortest-path (SP) kernel between $x$ and $x'$;  
$$
K_{SP}(x,x') = K^1_{\otimes}(y, y').
$$
 Lemma 3  in \cite{borgwardt2005shortest} showed that this kernel is  positive definite. 
 

\paragraph{Weisfeiler-Lehman Graph Kernels} Weisfeiler-Lehman Graph Kernels {have been}  proposed by \cite{shervashidze2011weisfeiler}; {these kernels are based on the Weisfeiler-Lehman test for graph isomorphisms and involve}   counting  matching subtrees between two given graphs. Theorem 3 in \cite{shervashidze2011weisfeiler} showed the positive definiteness of {these}  kernels.  {In our implementation, we adapted an} efficient implementation from the $\mathtt{graphkernel}$ package \citep{sugiyama2018graphkernels}. 

\section{Vector-Valued RKHS}\label{supp:vvRKHS}

{The general set-up for vector-valued RKHS for finite networks is as follows.} {Let $N = {n \choose 2}$  denote the index set of vertex pairs in a graph $x \in \{0,1\}^N$. For $s \in [N]$ let $x^{-s} \in \{0, 1\}^{N-1} =: {\mathcal X}^{-s} $ denote the collection of edge indicators except the one for $s$ and let $x_s \in \{0,1\}=: \mathcal{X}^{s} $ denote the edge indicator for $s$. } 
{When the underlying graph is random, we}  use  similar {notation}   $X^{-s}$, $X^{s}$ to denote the  corresponding random variables . 
For $s\in[N]$, let
$l_{s}:\mathcal{X}^{s}\times\mathcal{X}^{s}\rightarrow \mathbb{R}$
be 
reproducing kernels, with associated RKHS $\mathcal{H}_{l_{s}}$.
{Let} 
 $\varphi_{s}:x_{s}\in\mathcal{X}^{s}\mapsto l_{s}(\cdot,x^{s})\in\mathcal{H}_{l_{s}}$ denote the corresponding feature maps of $(l_{s})_{s\in[N]}$.

The RKHS kernels  $l_s$,  or those used in \cite{chwialkowski2016kernel} or \cite{ liu2016kernelized}, have scalar outputs, {while the} 
RKHS kernel $\ell_{-s}$ has an output in $\mathcal{L}(\mathcal{H}_{l_s})$, {the Banach space of bounded operators from $\H_{l_s}$ to $\H_{l_s}$;} 
we refer {to the  space $\H_{l_s}$} for $\ell_{-s}$ as a vector-valued RKHS (vvRKHS). 
All the kernels {used here are assumed to be}  positive definite and bounded.
As composition preserves positive definiteness, we then consider the kernel: $K: (\mathcal{X}^{s}\otimes \mathcal{X}^{-s}) \times (\mathcal{X}^{s}\otimes \mathcal{X}^{-s}) \to \R$, with associate RKHS $\H_K$. 

{{In our experiments we} assume}
that the $l_{s}$ corresponds to the same RKHS function: $l_{s} \equiv l, \forall s \in [N]$.
 We further assume the vvRKHS $\H_{\ell}$ has the form 
$$
\ell (x^{-s}, (x')^{-s'}) = k(x^{-s}, (x')^{-s'}) \mathbb{I}_{\H_l\times \H_l},
$$
where $\mathbb{I}_{\H_l\times \H_l}$ is the identity map from $\H_l$ to $\H_l$ and $k$ is the graph kernel of choice.
The RKHS  defined via composition reads
\begin{equation*}
K((x^s, x^{-s}), ((x')^{s'}, (x')^{-s'})) = k(x^{-s}, (x')^{-s'})l(x^{s}, (x')^{s'}).
\end{equation*}
{{For a single observed network $x$, as} $\H_{l}$, $\H_{\ell}$ are the same for all $s$, 
{it holds that} for  $s, s' \in [N]$:
\begin{equation*}
K((x^s, x^{-s}), (x^{s'}, x^{-s'})) = l(x^{s}, x^{s'}).
\end{equation*}
{In our implementation we use  the kernels 
$k(x^{-s}, \cdot) = k(x^{(s,1)}, \cdot) + k(x^{(s,0)}, \cdot) $ from Section B, 
defined not on the whole graph $x$ but on the set $x^{-s}$.}


\section{Additional Details on Distance-based Test Statistics}\label{supp:mgra_tv}
\subsection{Modified Graphical Tests with Total-Variation {Distance}}\label{supp:mgra_tv2}
\begin{figure*}[t!]
\begin{center}
		\subfigure[The null model]{\includegraphics[width=0.48\textwidth]{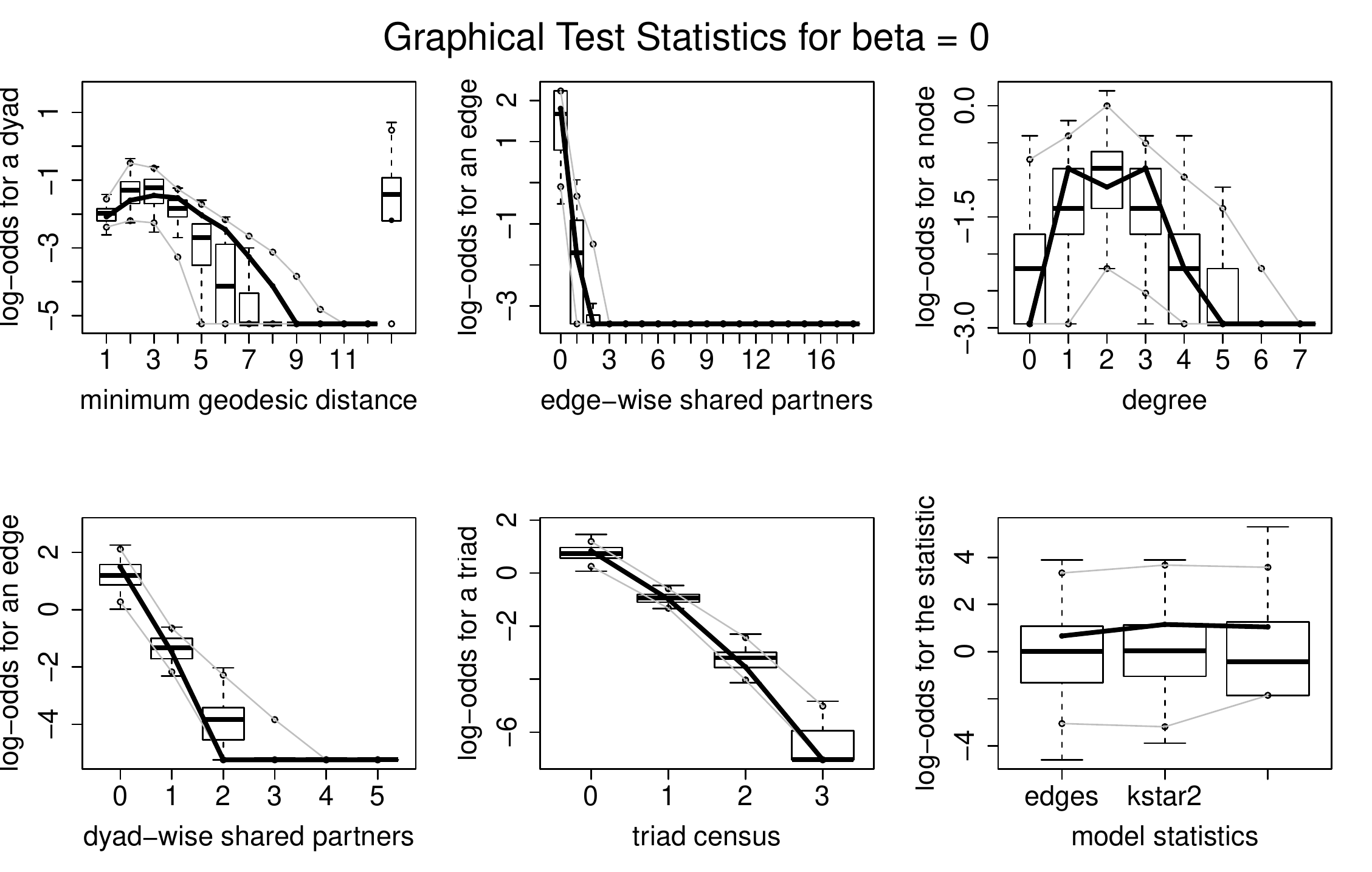}\label{fig:graphical_null}}\subfigure[A small perturbation {of the null model} 
		]{	\includegraphics[width=0.48\textwidth]{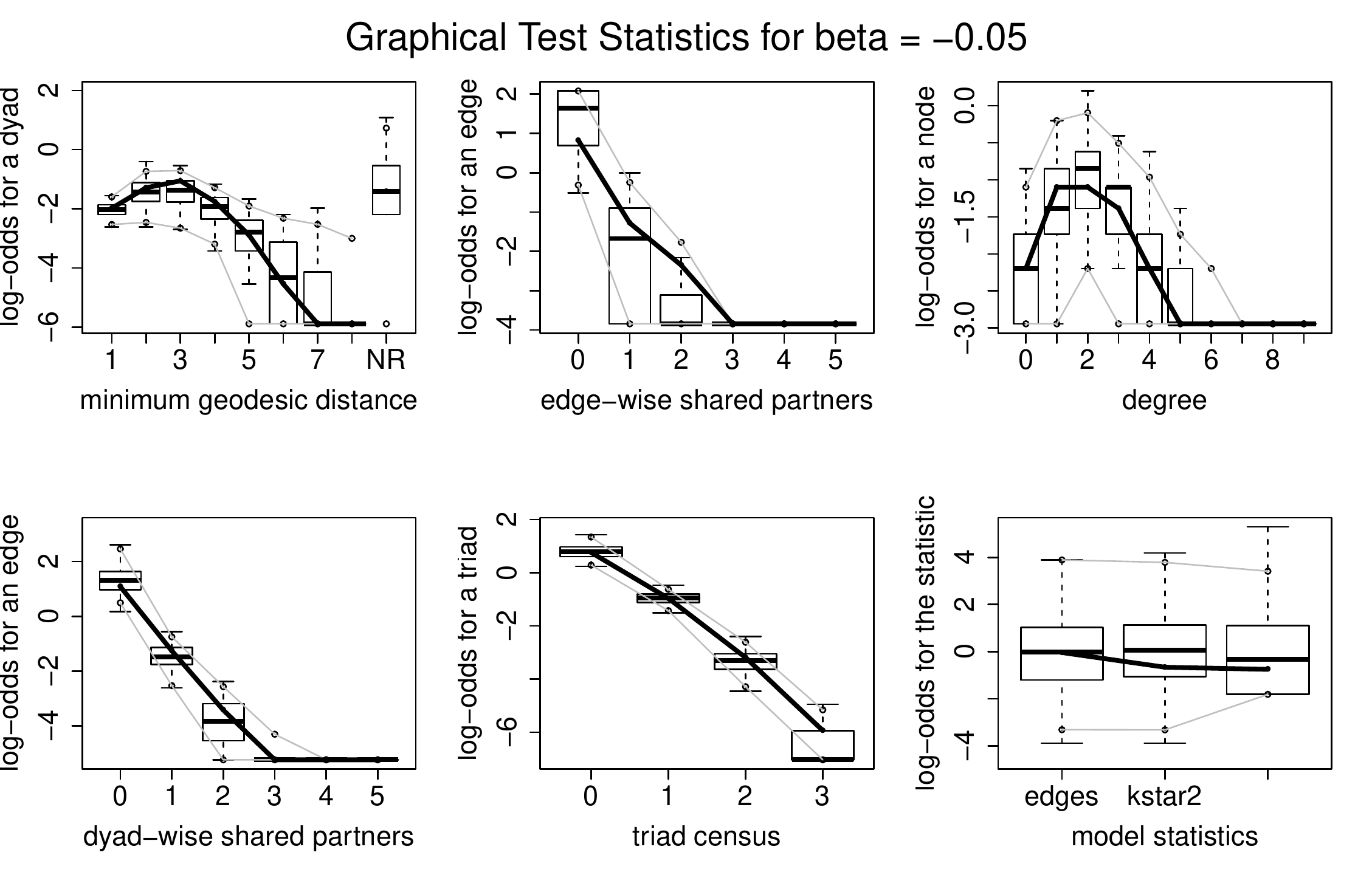}\label{fig:graphical_hard}}
		\subfigure[
		A {moderate} perturbation {of the null model} 
		]{
		\includegraphics[width=0.48\textwidth]{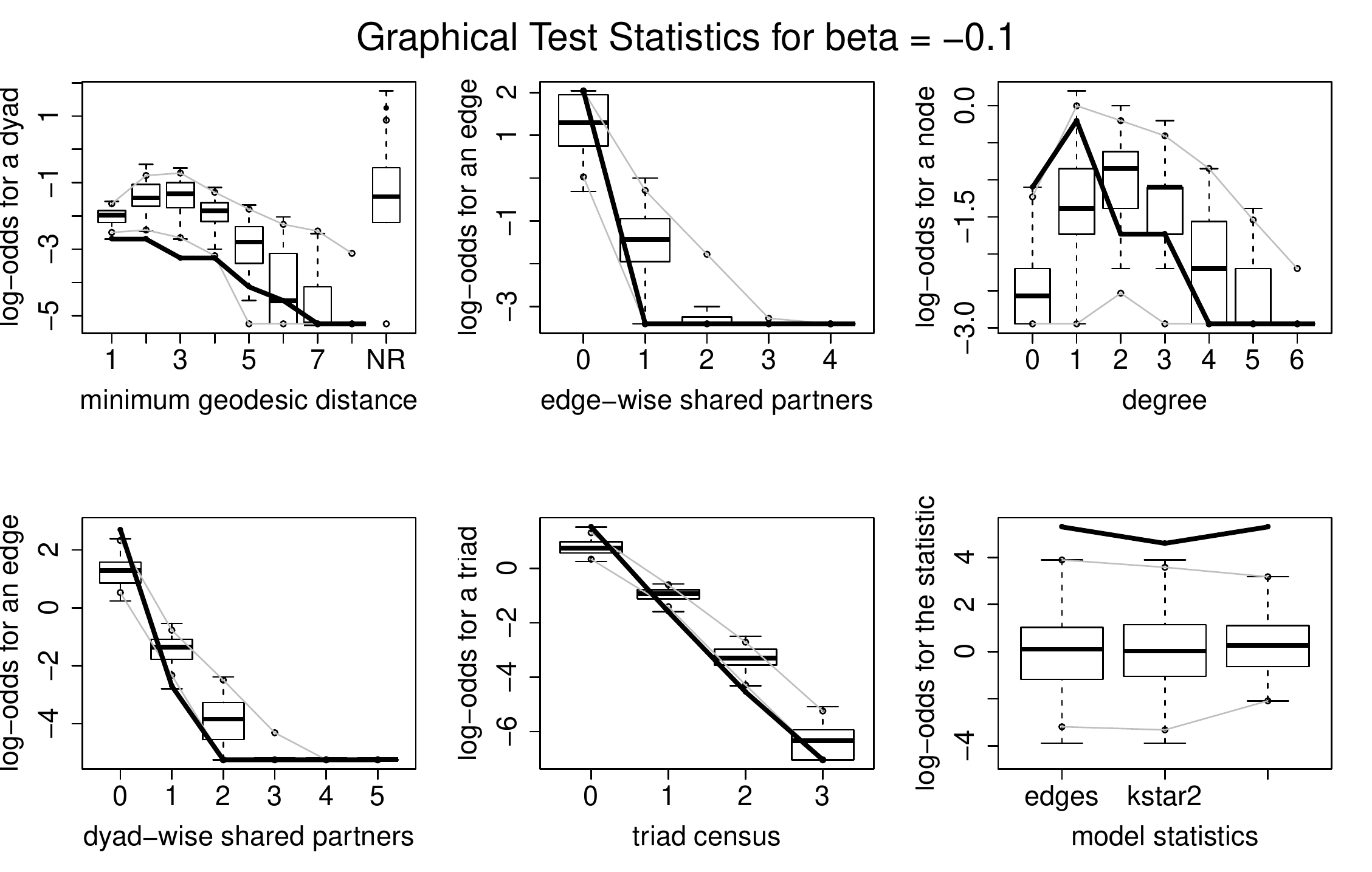}\label{fig:graphical_mod}}\subfigure[
		A {larger} perturbation {of the null model}]
	{\includegraphics[width=0.48\textwidth]{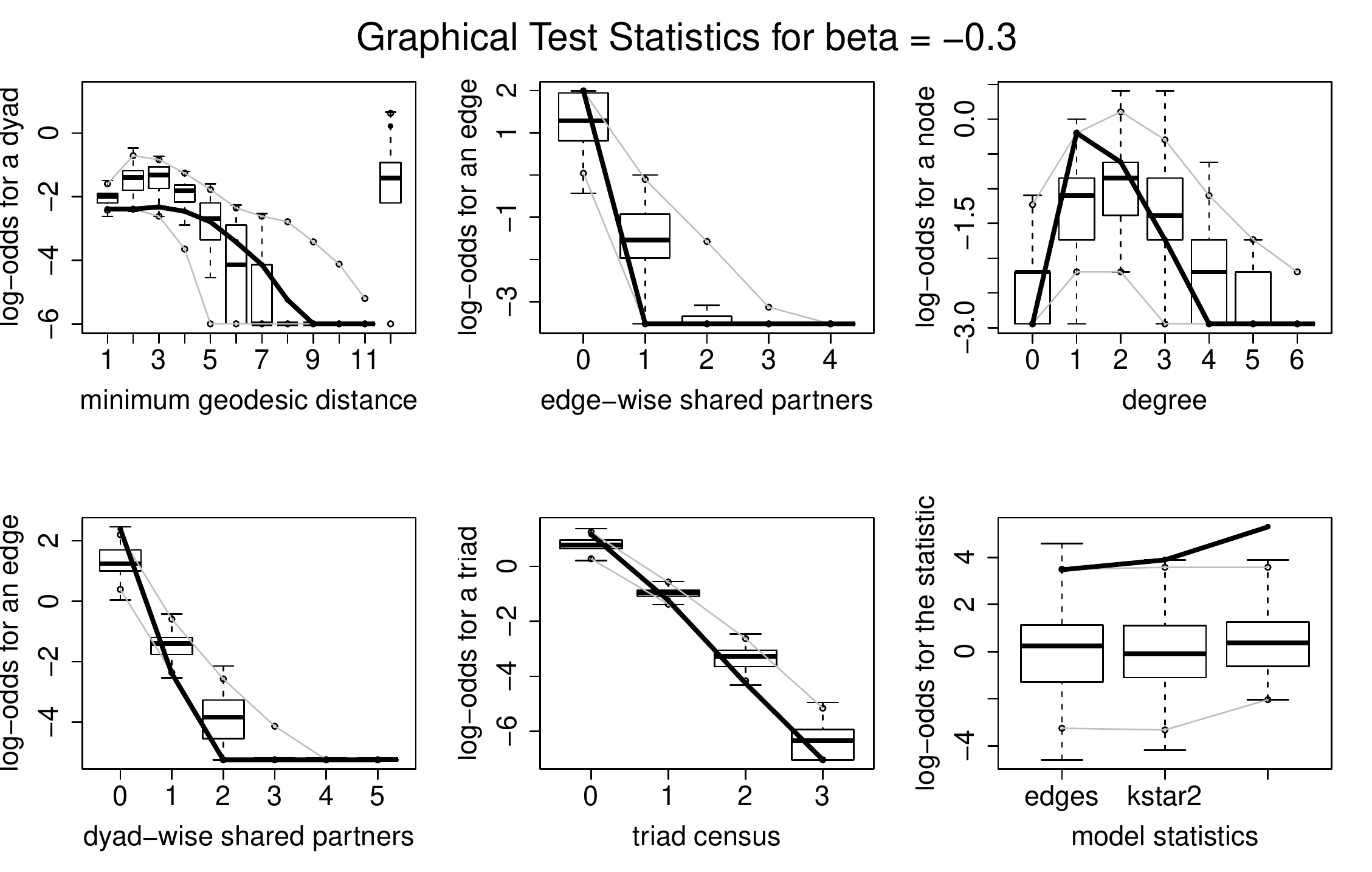}\label{fig:graphical_easy}}
	\caption{Graphical Tests with different beta parameters}\label{fig:graphical_test_demo}
	\end{center}
\end{figure*}

{Here we} give  the details of the modified graphical tests based on Total-Variation (TV) distance, \textbf{mGra},  presented in Section \ref{sec:exp}. {To assess the goodness-of-fit to a specific ERGM,} \cite{hunter2008goodness} proposed to compare network statistics from the observed network to {those of}  simulated networks from the null model {via box plots and Monte-Carlo $p$-values}. {These network statistics are 
\begin{itemize}
\item the degree distribution, with $d_k$ the number of vertices which have degree $k$;
\item the number of edge-wise shared partners, which is the number of pairs of vertices which are neighbours and which have exactly $k$ common neighbours; 
\item the number of dyad-wise shared partners,  which is the number of pairs of vertices which  have exactly $k$ common neighbours (but are not necessarily themselves neighbours);
\item the triad census, with 4 possible triads where triads are configurations on 3 vertices; the configurations are 0 edges, 1 edge, 2 edges and 3 edges;
\item the statistics which are included in the ERGM {as in Definition \ref{def:ergm}}.
\end{itemize} 
} Fig.~\ref{fig:graphical_test_demo} shows an example of a graphical test based on the E2ST model Eq.\eqref{eq:e2st} with the 2-star coefficient $\beta_2$ perturbed. By comparing whether the observed statistics (the bold line) deviates from the simulated null, one can visually {assess}  whether the null hypothesis should be rejected. 
For instance, in Fig.~\ref{fig:graphical_null} where the network is generated from the null {distribution}, the observed network statistics are all within the range in which 95 percent of the simulated observations fall.

{When} the difference between the null {distribution and the distribution which generates the data}
is small, the graphical {method}  may not easily distinguish the {two models} depending on the network statistics of choice. As shown in Fig.~\ref{fig:graphical_hard}, with a network from a model with small perturbation from the null {distribution}, we see {this}
effect.
However, when the difference between {data simulated under} the null  {distribution} and the data is {substantial} enough, we can see, e.g. from Fig.~\ref{fig:graphical_mod}, {that} the minimum geodesic distance and the triad census from the observed network clearly differ from the simulated null, {and} the null {hypothesis}  can be rejected. {The box plots are also used to carry out Monte Carlo tests for each possible observation (for example a {specific triad count}) by giving a $p$-value for this specific test.} 

{While every observed value can be used for a Monte Carlo test, \cite{hunter2008goodness} does not provide a}
systematic procedure to reach an overall conclusion about rejection.  For instance, it is not clear 
whether the null is {to be} rejected when Fig.~\ref{fig:graphical_easy} is observed. To surpass such issue, we further develop the testing procedure {by using the}
{TV} distance between distributions {for the observed and simulated distributions of the summary statistics} from \cite{hunter2008goodness}. Denote by  $S$  the random variable of a network statistic of choice and  by $\mathcal{S}$  the space for $S$. Using the vertex degree of a simple undirected  network on $n$  {vertices} as an example,  $S$ is a discrete random variable  taking values from $0$ to $n-1$. Further denote by $S_{z'}$  the network statistic of an observation $z'$ from the null model $q$ and by $s_x$ the network statistics from the observed $x$. {Then with ${\mathcal R}$ denoting the set of possible values of $S$,}
$$
d_{TV}(S_{z'}; S_x) = \sup_{A \subset {\mathcal R} } | \E[h_A(S_{z'}) - h_A(S_x)] |  = \frac{1}{2}\sum_{s\in\mathcal{S}}\left|P(S_{z'}=s)- P(S_x=s) \right|$$
where $h_A(s) = \mathds{1}_{s\in A}$ is
{1 if $s \in A$ and 0 otherwise.}
Our test statistic measures the distance between the distribution of a network statistic $S$ in the observed network  $x$ and  under the null model $q$ as follows:
$$
D_{TV}(q,x;S) = \E_{z'\sim q} [d_{TV}(S_{z'}; S_x)] .
$$
To {estimate}  $\E_q$, we simulate $m'$ networks from the null model $q$, i.e. $z'_1,\dots, z'_{m'}\sim q$ {and use as} empirical estimate for $D_{TV}(q,x;S)$ 
$$
\widehat{D_{TV}}(q,x;S) = \frac{1}{m'}\sum_{j=1}^{m'} [d_{TV}(S_{z'_j}; S_x)].
$$
To {assess} how the test statistics is distributed under the null {hypothesis}, i.e. $x\sim q$, we simulate $z\sim q$ {from the null distribution}. Similar to a Monte-Carlo test, we simulate {independent} network samples $z_1, \dots, z_m \sim q$ and compute
$\widehat{D_{TV}}(q,z_i;S)$, for $i\in [m]$. Then we reject the null if $\widehat{D_{TV}}(q,x;S)$ exceeds the $(1-\alpha)$-quantile level in the { simulated} {observations $\{\widehat{D_{TV}}(q,z_1;S),\dots,\widehat{D_{TV}}(q,z_m;S)\}$ under the}  null {distribution}.

\subsection{Test Statistics with Mahalanobis {Distance}}
{\cite{lospinoso2019goodness} proposed using a Mahalanobis distance instead of the total variation distance. Suppose that a vector $S(x)$ of network summaries is observed and that the null distribution is parametrised by $\theta$. Denote $\mu(\theta) = \E_\theta (X)$ as the expectation  and $\Sigma(\theta) = Cov_\theta (X)$ as the covariance matrix under $\theta$. The {\it Mahalanobis distance}
$$  D_M(x, \theta;S) = ( S(x) - \mu (\theta) )^{\top} \Sigma (\theta)^{-1}  ( S(x) - \mu (\theta) )
$$
can then be used as test statistic. In practice, $\mu(\theta)$ and $\Sigma(\theta) $ are estimated using {independent} simulations $x_k, k=1, \ldots, m, $ 
from the distribution specified by $\theta$;
\begin{eqnarray*}
{\widehat \mu} = \frac1m \sum_{k=1}^m S (x_k); \quad \quad  \quad \quad \quad \quad 
{\widehat \Sigma} = \frac1m \sum_{k=1}^m ( S (x_k) - {\widehat \mu} ) ( S (x_k) - {\widehat \mu} )^{\top};
\end{eqnarray*} 
\begin{eqnarray*} 
{\widehat{D_M}} (x) = ( S(x) - {\widehat \mu})^{\top} {\widehat \Sigma} ^{-1}  ( S(x) - {\widehat \mu}) .
\end{eqnarray*} 
The $p$-value of the test is estimated by the plug-in estimator 
 $${\widehat p} = \frac1m \sum_{k=1}^m {\mathbb{1}} \{ {\widehat{D_M}} (x_k) > {\widehat{D_M}} (x)\}.$$
 In the main text this approach is abbreviated {\bf{MD}} and applied to the degree distribution for ERGMs.}

\begin{figure*}[t!]
	\begin{center}
    		\subfigure[$n=20$, $\alpha=0.05$,  {with} $\beta_2$  {in Eq.\eqref{eq:e2st}}
		 perturbed]{\includegraphics[width=0.47\textwidth, height=0.34\textwidth]{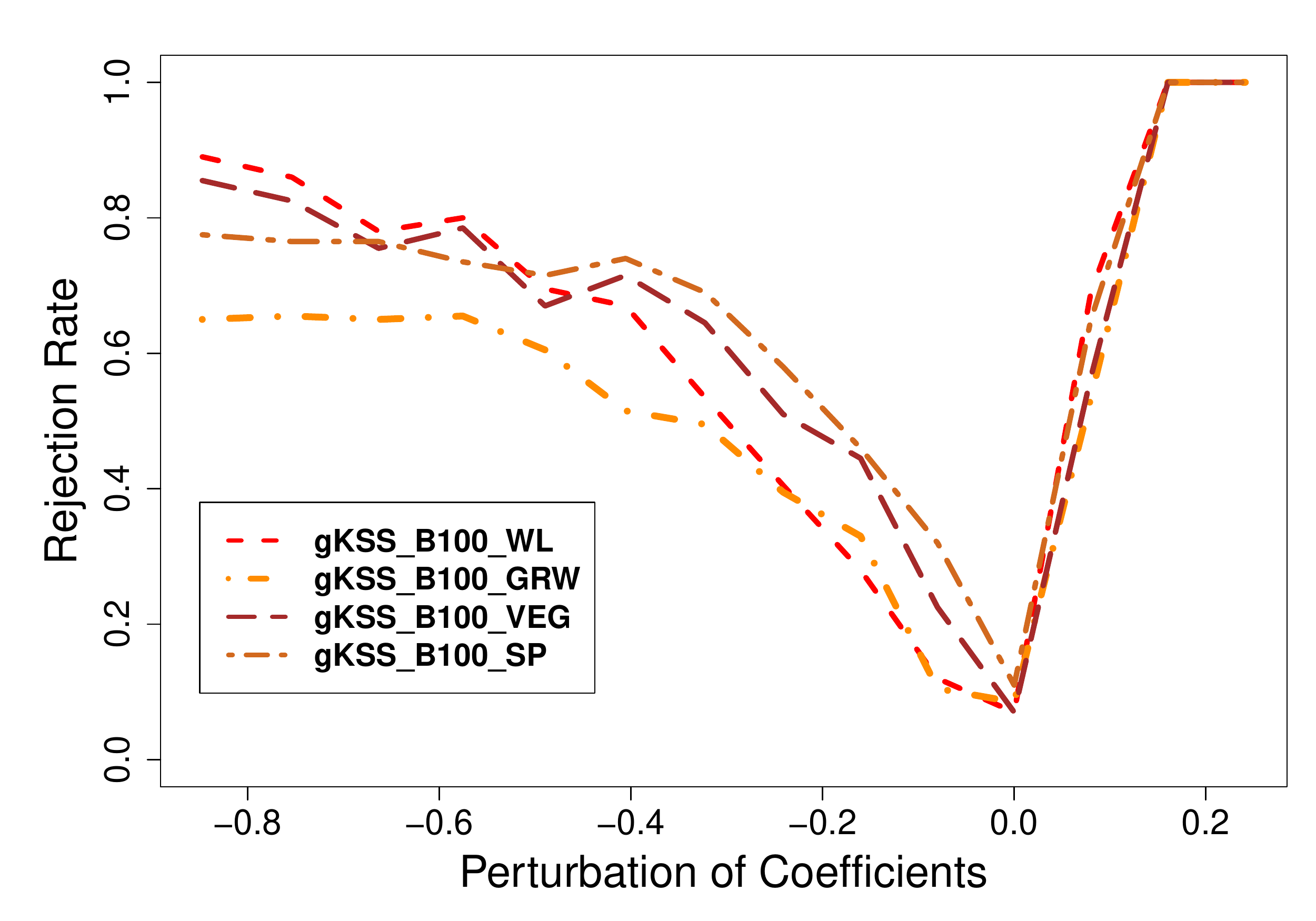}\label{fig:perturb_kernels}}
        \subfigure[{log computational} time  {for one test,} in seconds, {with $m = 1000$ simulated networks}]{\includegraphics[width=0.47\textwidth, height=0.34\textwidth]{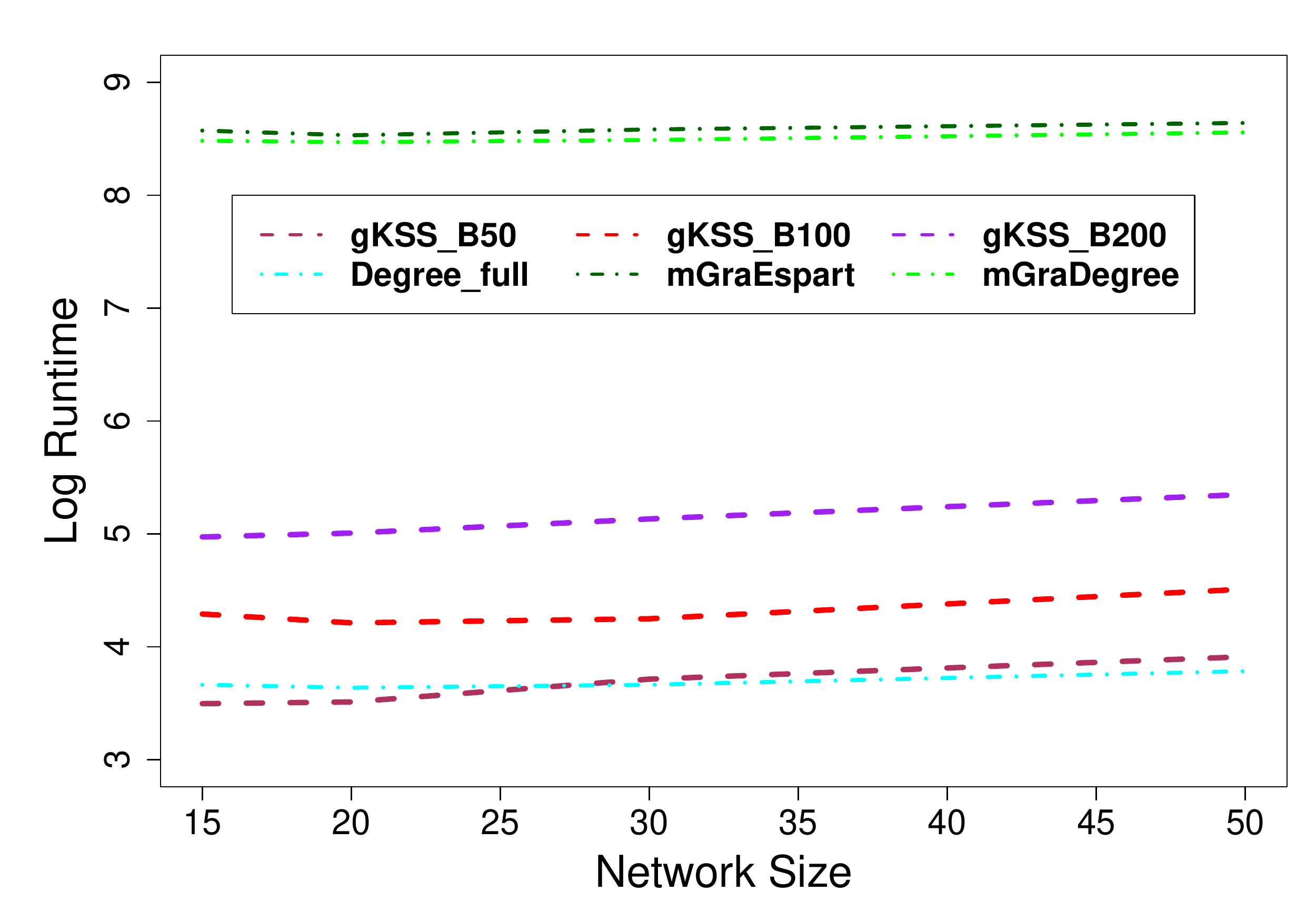}\label{fig:runtime}}
	\caption{Additional experiment results}\label{fig:exp_add}
	\end{center}
\end{figure*}

\section{Additional Experiment Results}\label{supp:exp}
\paragraph{Test performances with graph kernels}
Fig.\ref{fig:perturb_kernels} shows the results  for testing the E2ST model Eq.\eqref{eq:e2st} with the 2-star coefficient $\beta_2$ perturbed using the different kernels described in Section {B}. {Using the abbreviations from Section B, the relevant choices of kernel parameters are} $\sigma =1$ for the VEG kernel, level$=5$ for the WL kernel, and $\lambda = \frac13$ in the GRW kernel.  
Similar to the WL kernel used in the main text
, the other choices of graph kernels achieve {fairly} good test power with the gKSS statistic. {In our additional experimental results on the rejection rate, the} re-sample size is $B=100$ for all kernel choices. From  Fig.\ref{fig:perturb_kernels} we see that the test power is slightly higher with a small perturbed coefficient when the SP kernel and  the VEG kernel are employed,  while for larger perturbed coefficient ({resulting in} sparser graphs) the WL kernel  better distinguishes the observation from the null. {For large negative $\beta_2$ the GRW kernel has the poorest rejection rate. These differences in performance are no surprise as} 
different choices of kernel {emphasise} 
different aspect of graph topology.

\paragraph{Computational time} In Fig.\ref{fig:runtime}, we give more results for the computational time {for one test, with 1,000 simulated networks}. {These results complement the reported results of Table 1 in the main text.} As the {number of vertices in the network} increases, there is an increase in the computational complexity. However, as the main computation costs come from simulating the ERGM, we see from the plot that the slope is not {substantial} compared {to the difference in testing procedures}.

\section{{Comparison with the Kernel Discrete Stein Discrepancy on Testing Goodness-of-fit}
}\label{supp:KDSD_supp}

\subsection{Discrete Stein Operator}
In this section, {we}  {compare our approach }with the discrete Stein operator {introduced in \cite{yang2018goodness}. First we need some definitions.} 
\begin{Definition}\label{def:cyclic_perm}[Definition 1 \citep{yang2018goodness}](Cyclic permutation). For a set {${\mathcal{N}}$} of finite cardinality,
a cyclic permutation $\neg: {{\mathcal{N}}} \to {{\mathcal{N}}}$ is a bijective function
such that for some ordering $x^{[1]}, x^{[2]}, \dots , x^{[|X |]}$ of the
elements in ${{\mathcal{N}}}$ , $\neg x^{[i]} = x^{[(i+1) mod |X |]}, \forall i = 1, 2, \dots , |X |$.
\end{Definition}

\begin{Definition}\label{def:difference_operator}[Definition 2 \citep{yang2018goodness}]
Given a cyclic permutation $\neg$ on {${\mathcal{N}}$} , for any d-dimensional vector $x = (x_1, . . . , x_d)^{\top} \in {\mathcal{N}}^d$, write $\neg_i x := (x_1, \dots , x_{i-1}, \neg x_i, x_{i+1}, \dots, x_d)^{\top}.$ 
For any function $f : {\mathcal{N}}^d \to \mathbb{R}$, denote the (partial) difference operator as
\begin{align*}
    \Delta_{x_i}f(x): = f(x) - f(\neg_i x), \quad \quad  i=1,\dots,d
\end{align*}
and {introduce} the difference operator:
\begin{align*}
    \Delta_\neg f(x): = ( \Delta_{x_1} f(x), \dots ,  \Delta_{x_d} f(x))^{\top}.
\end{align*}
\end{Definition} 
{Here we use the notation $ \Delta_\neg$ to distinguish it from the notation in the main text, where we used  
$\Delta_s h(x) = h(x^{(s,1)}) - h(x^{(s,0)})$ and $ || \Delta h ||  = \sup_{s \in [N]} | \Delta_s h(x)|.$
}

For discrete {distributions $q$}, \cite{yang2018goodness} proposed {the following} discrete  Stein operator, which is based on the difference operator {$\Delta_\neg $} constructed from  a cyclic permutation: 
\begin{align}
\A^{D}_q f(x) = f(x)\frac{\Delta_\neg q(x)}{q(x)} - \Delta_\neg^{\ast} f(x), \label{eq:disrete_stein}
\end{align}
where $\Delta_\neg^{\ast}$ denotes the adjoint operator of $\Delta_\neg$.

{In particular, for $q$ the distribution of an ERGM, {with ${\mathcal{N}} = \{0, 1\}^{N},$} the discrete Stein operator proposed \citep{yang2018goodness}  can be written in the form of $\A^{D}_q f(x) = \sum_s \A^{D,s}_q f(x)$ where} 
\begin{equation}\label{eq:kdsd_component}
\A^{D,s}_q f(x)=(-1)^{\mathds{1}_{\{x=x^{(s,0)}\}}}\frac{f(x^{(s,1)})q(x^{(s,0)}) - f(x^{(s,0)})q(x^{(s,1)})}{q(x)}    .
\end{equation}

Recall the ERGM Stein operator of the form $\A_{q}f(x) = \frac{1}{N}\sum_{s\in[N]}\A_q^{(s)}f(x)$ and  Eq.\eqref{eq:stein_component},
\begin{align*}
    \A^{(s)}_q f(x) &= q(x^{(s,1)}|x) \Delta_s f(x) + \left( f(x^{(s,0)}) - f(x)\right) \\
    &= \frac{q(x^{(s,1)})}{q(x^{(s,1)})+q(x^{(s,0)})} \left(f(x^{(s,1)}) - f(x^{(s,0)}) \right) + \left( f(x^{(s,0)}) - f(x)\right)\\
    & = \frac{\mathds{1}_{\{x=x^{(s,0)}\}}q(x^{(s,1)}) - \mathds{1}_{\{x=x^{(s,1)}\}}q(x^{(s,0)})}{q(x^{(s,1)})+q(x^{(s,0)})} \left(f(x^{(s,1)}) - f(x^{(s,0)}) \right).
\end{align*}

{We} illustrate the difference between the ERGM Stein operator and the discrete Stein operator 
for a 
Bernoulli random graph
with  $\P (s=1) = q, \forall s $. Due to the independence, we have $q(x^{(s,1)}|x) = q$ and $q(x^{(s,0)}|x) = 1-q$. With  Eq.\eqref{eq:stein_component}, our Stein operator becomes
\begin{equation}\label{eq:ER_graph_stein}
\A_q f(x) =\frac{1}{N} \sum_s \left({q} - x_s \right)(f(x^{(s,1)})-f(x^{(s,0)})).
\end{equation}
The $\operatorname{KDSD}$ in this case can be written as:
\begin{align*}
\A_q^{D} f(x) &=\frac{1}{{q(x)} }\sum_s(-1)^{1-x_s} \left((1-q)f(x^{(s,1)})-q f(x^{(s,0)})\right)
\end{align*}
{with $q(x) = q^{\sum_s x_s} (1-q)^{N - \sum_s x_s}$.}
{Thus,}
for different values, Eq.\eqref{eq:ER_graph_stein} {is a weighted sum of the terms $ \left( f(x^{s,1}) - f(x^{s,0})\right)
$, while  $\operatorname{KDSD}$ is a weighted sum of the terms $\left( (1-q)f(x^{s,1}) - q f(x^{s,0})\right)$ and requires the calculation of the binomial probability $q(x)$.} 

{The operators in Eq.\eqref{eq:kdsd_component} and Eq.\eqref{eq:stein_component} clearly differ in their scaling as well as in their repercussions for re-sampling. While the operator in   Eq.\eqref{eq:stein_component}  emerges from Glauber dynamics and hence has a natural re-sampling interpretation, no such interpretation is available for the operator in  Eq.\eqref{eq:kdsd_component}.} 
Explicitly, the discrete Stein operator {${\mathcal{T}}_q^D$} has  $q(x)$ in the denominator, indicating the fixed $x$ realisation; however,  the Stein ERGM operator {${\mathcal{T}}_q$} has $q(x^{-s})$ in the denominator which {stems from} the conditioning in Glauber dynamics. 
{Consequently, the corresponding Stein discrepancy (called KSDS) differs from Eq.\eqref{eq:ksd} in the main text, and,  although usually only one network is available, the goodness-of-fit test in \cite{yang2018goodness} requires independent and identically distributed network observations.} 

{A second key difference is that  {the test in} \cite{yang2018goodness} requires the support of the unknown network distribution to be identical to the support of the ERGM which is described by $q$. In practice this condition is difficult  if not impossible to verify. In contrast, $\operatorname{\widehat{gKSS}}$ does not make any such assumption.}

\subsection{Comparison Between Graph Testing}

\paragraph{Testing with multiple graph observations}
The relevant kernel discrete Stein discrepancy {(KDSD)} from the discrete Stein operator \citep{yang2018goodness} is defined via taking the supreme over appropriate unit ball RKHS test functions, similar as in Eq.\eqref{eq:ksd}
\begin{equation}
\operatorname{KDSD}(q\|p;\H) = \sup_{\|f\|_{\H}\leq 1}\E_p[\T_q^{D} f(x)].
\label{eq:kdsd}
\end{equation}
{\cite{yang2018goodness} built a} goodness-of-fit testing procedure 
{based on the} KDSD  for ERGM 
for multiple graph observations. 
Let $x_1,\dots,x_m \sim p$, be $m$ {independent} identically distributed graph observations. The KDSD 
{is} empirically estimated from the observed samples;  and as the number of observed samples $m\to \infty$, {in probability,}
$$\frac{1}{m} \sum_i [\T_q^{D} f(x_i)]  \to \E_p[\T_q^{D} f(x)].$$
The rejection threshold 
{is}
determined via a wild-bootstrap procedure \citep{chwialkowski2014wild}.

While the $\operatorname{\widehat{gKSS}}$ type of statistics} based on the ERGM Stein operator in Eq.\eqref{eq:ergm_stein},  $\T_q f(x) = \frac{1}{N}\sum_{s\in [N]} \T_q^{(s)}f(x)$, {focuses}
on a single graph observation,  this ERGM Stein operator {could {similarly} be used to assess goodness-of-fit  when multiple graph observations are available.} {In particular,  $\E_q[\T_q f(x)] = 0$.}
 {Hence,  we {introduce}
 the graph kernel Stein discrepancy (gKSD) {as} 
$$\operatorname{gKSD}(q\|p;\H) = \sup_{\|f\|_{\H}\leq 1}\E_p[\T_q f(x)]= \sup_{\|f\|_{\H}\leq 1}\E_p \left[\frac{1}{N}\sum_s \T^{(s)}_q f(x)\right].$$
{Here the sum is taken over all $N$  pairs of vertices and the expectation is taken with respect to the ERGM $q$. When there are $m$ independent observations $x_1, \ldots, x_m \sim p$ available then we} can empirically estimate $\operatorname{gKSD}(q\|p;\H)$ by 
$\frac{1}{m}\sum_i [\frac{1}{N}\sum_s \T^{(s)}_q f(x_i)]$, which is {weakly} consistent {by} the law of large numbers. Then we   {use this statistic to} build a  goodness-of-fit test for multiple graph observations, {determining the threshold via the same wild-bootstrap procedure as for the KDSD}.

{To compare the KDSD and the $\operatorname{gKSD}$ tests we} consider the goodness-of-fit test setting as studied in \cite{yang2018goodness}, using the E2ST model as presented in Eq.\eqref{eq:e2st}. We set the null parameters $\beta$ to 
$(\beta_1,\beta_2,\beta_3) = (-2, 0.0, 0.01)$ {and carry out a test at significance level $\alpha = 0.05$ using 100 repeats}. For the alternative, we perturb the coefficient for 2-stars, $\beta_2$, and report the rejection rate in Table \ref{tab:E2ST20WL}. Note that $\beta_2=0.00$ recovers the null {distribution}.  {In this experiment,}
with a small number of graph observations, $m=30$, gKSD captures the difference between the null model and the alternative model better, resulting in a higher test power, compared to KDSD. {Both gKSD and KDSD have higher power for $\beta_2 > 0$ than for $\beta_2 < 0$ of the same magnitude. This finding is plausible as increasing $\beta_2$ leads to denser networks.} 

\begin{table}[ht]
\begin{center} 
\resizebox{0.9\textwidth}{!}{%
\begin{tabular}{rrrrrrrrrrrr}
  \hline
{$\beta_2$} & -0.1 & -0.08 & -0.06 & -0.04 & -0.02 & 0.00 & 0.02 & 0.04 & 0.06 & 0.08 & 0.1 \\ 
  \hline 
gKSD  & 0.32 & 0.30 & 0.24 & 0.14 & 0.10 & 0.04 & 0.08 & 0.22 & 0.18 & 0.28 & 0.54 \\
KDSD   & 0.08 & 0.05 & 0.6 & 0.04 & 0.01 & 0.02 & 0.03 & 0.03 & 0.06 & 0.12 & 0.16 \\
   \hline
\end{tabular}
}
\end{center}
\caption{Rejection rate for the E2ST model $(\beta_1, \beta_2, \beta_3) = (-2, 0, 0.01)$ with perturbation of the 2-star coefficient $\beta_2$: W.L. Kernel of level 3; sample size $m=30$; graph size $n=20$; test significance level $\alpha = 0.05$.}
\label{tab:E2ST20WL}
\end{table}

\paragraph{Testing with a single graph observation}
The ERGM Stein operator satisfies {the mean zero property Eq.\eqref{meanzeroproperty}}
when flipping each edge $s$ given the rest of the graph $x_{-s}$. 
This is {a} key ingredient that KDSD does not {satisfy};  
KDSD relies on a cyclic permutation as in Definition \ref{def:cyclic_perm} to construct the partial difference operator in Definition \ref{def:difference_operator}, which depends on the order sequence of the cyclic permutation. As such, the {mean zero property of their} Stein operator {is based on sign flips} 
in each state of the discrete variable, instead of flipping the edge probability. Thus, the discrete Stein operator $\T^D_q$ {could not easily be}  adapted to construct a {subsampled} Stein statistic {such as $\operatorname{\widehat{gKSS}}$} to perform goodness-of-fit testing with a single graph observation.

\paragraph{Testing with {a few} graph observations}
An interesting setting {which is related to that of a}
single graph observation, is that {a few} graphs are observed, {with} the number of graphs assumed to be finite and  not tending infinity {with network size}. With the proposed gKSD goodness-of-fit test for a single graph observation, a possible approach and  a potential future research direction is applying multiple tests of goodness-of-fit, {one for each observed network,} with a Bonferroni correction \citep{bonferroni1936teoria} {to correct for multiple testing}. 

\end{document}